\documentclass[intlimits,twoside,a4paper]{article}

\usepackage{amsthm}
\usepackage[cp1251]{inputenc}

\usepackage[eqsecnum]{cmpj3}

\usepackage{bm}

\theoremstyle{plain} 
\newtheorem{theorem}{Theorem}[section]
\newtheorem{proposition}{Proposition}[section]

\theoremstyle{definition}
\newtheorem{definition}{Definition}[section]

\theoremstyle{remark}
\newtheorem{remark}{Remark}[section]

\issue{2019}{22}{3}{33101}
\doinumber{10.5488/CMP.22.33101}

\title[Current algebra representations of integrable  factorized Schr\"{o}dinger type models]%
{The current algebra representations of quantum many-particle Schr\"{o}dinger type Hamiltonian models, their factorized structure and integrability%
\thanks{To  memory of  Nikolai N. Bogolubov, a mathematical physics giant of the XX-th century --- on his 110-th Birthday Jubilee.}}
\author[D. Prorok, A.K. Prykarpatski]{D. Prorok\refaddr{label1},
        A.K. Prykarpatski\refaddr{label2}}
\addresses{
\addr{label1} Department of Physics and Applied Computer Science, AGH University of Science and Technology, \\Krakow, Poland
\addr{label2} Institute of Mathematics, Department of Physics,
Mathematics and Computer Science, Cracow University of Technology,
Krakow 31-155, Poland
}

\date{Received July 16, 2019, in final form August 6, 2019}

\DeclareMathOperator{\Diff}{Diff}
\DeclareMathOperator{\Dom}{Dom}
\DeclareMathOperator{\per}{per}
\DeclareMathOperator{\End}{End}
\DeclareMathOperator{\tr}{tr}
\DeclareMathOperator{\reg}{reg}
\DeclareMathOperator{\supp}{supp} 

\begin{document}

\maketitle

\begin{abstract}
There is developed a current algebra representation scheme for
reconstructing algebraically factorized quantum Hamiltonian and symmetry
operators in the Fock type space and its application to quantum Hamiltonian
and symmetry operators in case of quantum integrable spatially many- and
one-dimensional dynamical systems. As examples, we have studied in detail the
factorized structure of Hamiltonian operators, describing such quantum
integrable spatially many- and one-dimensional models as generalized
oscillatory, Calogero-Sutherland, Coulomb type and nonlinear Schr\"{o}dinger
dynamical systems of spinless bose-particles.

\keywords Fock space, current algebra representations, Hamiltonian reconstruction, Bogolubov generating functional, Calogero-Moser-Sutherlan model, quantum integrability, quantum symmetries
\pacs  11.10.Ef, 11.15.Kc, 11.10.-z, 11.15.-q, 11.10.Wx, 05.30.-d
\end{abstract}

\section{Introduction}

Given any classical Hamiltonian many-particle non-relativistic system with
the standard cotangent phase space $T^{\ast }(\mathbb{R}^{3})^{\otimes N},$
where the quantity of particles  $N\in \mathbb{Z}_{+}$ is fixed, there
is a standard recipe for producing a quantum system by a method known as
\textquotedblleft\textit{canonical quantization}\textquotedblright,
assigning to the system a suitably constructed \cite{Takh} self-adjoint
Hamiltonian operator, acting in the related Hilbert space $\mathcal{H}=L_{2}(%
\mathbb{R}^{3\otimes N};\mathbb{C}).$ In case when all the particles are
equivalent to each other and the particle number $N\in \mathbb{Z}_{+}$ can
vary within the system, there is another recipe applied to this system,
called the ``\textit{second quantization}'', producing the corresponding \cite%
{Bere,BoBo} quantum self-adjoint Hamiltonian operator, acting already in a
specially constructed Fock space $\Phi _\text{F},$ whose basis vectors are
generated by means of actions of additional so-called ``\textit{creation}''
and ``\textit{annihilation}'' operators on a uniquely defined ``\textit{vacuum}'' 
zero-particle vector state $|0)\in \Phi _\text{F},$ whose structure in most
practical cases is hidden. Even though this method appeared to be very
effective for the study of numerous quantum many-particle Hamiltonian systems, some
important problems related to the \textit{a priori} non-self-adjointness
of the ``\textit{creation}'' and ``\textit{annihilation}'' operators in the
Fock space $\Phi _\text{F},$ urged researchers to suggest a dual
quantization scheme, based strictly on physically ``\textit{observable''
operators in a suitably constructed cyclic Hilbert space} $\Phi ,$ generated
by means of the so-called ``groundstate'' vector $|\Omega )\in \Phi ,$ and
being completely different from the Fock space $\Phi _\text{F}.$

Several authors have been developing this idea of quantizing nonrelativistic
models, making use of the \textit{local current algebra} operators \cite%
{Aref,Gold,Gold-1,GoGrPoSh,GoMeSh-1,GoMeSh-2,PaScWr} as the basic dynamical
variables, that is the density $\rho (x):\Phi \rightarrow \Phi $ and current
$J(x):\Phi \rightarrow \Phi $ operators at spatial point $x\in \mathbb{R}%
^{3},$ representing, as is well know, generators of the fundamental physical
symmetry group $\Diff(\mathbb{R}^{3})\ltimes \mathcal{S}(\mathbb{R}^{3};%
\mathbb{R})$, the semidirect product of the diffeomorphism group $\Diff(\mathbb{R}^{3})$ of the space $\mathbb{R}^{3}$ and the Schwarz space of
smooth real valued functions on $\mathbb{R}^{3}.$ Moreover, the
corresponding quantum Hamiltonian operators of the Schr\"{o}dinger type in
the Hilbert space $\Phi ,$ as it appeared to be very surprising,  possess
 a very nice factorized structure, completely determined by this
groundstate vector $|\Omega )\in \Phi .$ This fact posed a very interesting
and important problem of studying the related mathematical structure of
these factorized operators and the correspondence to the  
classical Hamiltonian many-particle non-relativistic systems generating them, specified by
some kinetic and inter-particle potential energy.

Analytical studies in modern mathematical physics are strongly based on the
exactly solvable physical models which are of great help in 
understanding  their mathematical and frequently hidden physical nature.
Especially, the solvable models are of great importance in quantum many-particle physics, amongst which one can single out  the oscillatory
systems and Coulomb systems, modelling phenomena in plasma physics, the
well known Calogero-Moser and Calogero-Moser-Sutherland models, describing a
system of many particles on an axis, interacting pair-wise through long
range potentials, modelling both some quantum-gravity\ and fractional
statistics effects. In this work we developed investigations of local
quantum current algebra symmetry representations in suitably renormalized
representation Hilbert spaces, suggested and devised before by G.A. Goldin
with his collaborators, having further applied their results to
constructing the related factorized operator representations for
secondly-quantized many-particle integrable Hamiltonian systems. The main
technical ingredient of the current algebra symmetry\ representation
approach consists\ in the weak equivalence of the initial many-particle
quantum Hamiltonian operator to a suitably constructed quantum Hamiltonian
operator in the factorized form, strictly depending only on its ground state
vector. The latter makes it possible to reconstruct the initial quantum
Hamiltonian operator in the case of its strong equivalence to the related
factorized Hamiltonian operator form, thereby constructing,\ as a
by-product, the corresponding\ $N$-particle groundstate vector for arbitrary
$N\in \mathbb{Z}_{+}.$ Being uniquely defined by means of the Bethe
groundstate vector representation in the Hilbert space, the analyzed
factorized operator structure\ of quantum completely integrable
many-particle Hamiltonian systems on the axis proves to be closely related
to their quantum integrability by means of the quantum inverse scattering
transform. As examples we have studied in detail the factorized structure
of Hamiltonian operators, describing such quantum integrable spatially
many- and one-dimensional models as generalized oscillatory,
Calogero-Moser-Sutherland and nonlinear Schr\"{o}dinger dynamical systems of
spinless bose-particles.

\section{\label{Sec_FQ1}Fock type Hilbert space, nonrelativistic quantum
current algebra and its representations}

\subsection{Preliminaries}

Let $\Phi $ be a separable Hilbert space, $F$ be a topological real linear
space and $\mathcal{A}:=\left\{ A(f):f\in F\right\} $ a family of commuting
self-adjoint operators in $\Phi $ (i.e., these operators commute in the sense
of their resolutions of the identity) with dense in $\Phi $ domain $\Dom$
$A(f):=D_{A(f)},f\in F.$ Consider the Gelfand rigging \cite%
{Bere-1,BeSh,GeVi} of the Hilbert space ${\Phi ,}$ i.e., a chain
\begin{equation}
\mathcal{D}\subset {\Phi }_{+}\subset {\Phi }\subset {\Phi }_{-}\subset
\mathcal{D}^{\prime },  \label{FQeq2.1}
\end{equation}%
in which ${\Phi }_{+}$ is a Hilbert space, topologically (densely and
continuously) and quasi-nucleus (the inclusion operator $i:{\Phi }%
_{+}\rightarrow {\Phi }$ is of the Hilbert-Schmidt type) embedded into
${\Phi },$ the space $\ {\Phi }_{-}$ is the dual to ${\Phi }_{+}$ as the
completion of functionals on ${\Phi }_{+}$ with respect to the norm $%
||f||_{-}:=\sup_{||u||_{+}=1} |(f|u)_{\Phi }|,$ $f\in {\Phi ,}$
a linear dense in $\mathcal{\Phi }_{+}$ topological space $\mathcal{%
D\subseteq }$ $\mathcal{\Phi }_{+}$ is such that $\mathcal{D\subset }$ $%
D_{A(f)}\subset {\Phi }\ $\ and the mapping $A(f):$ $\mathcal{D}$ $\mathcal{%
\rightarrow }$ ${\Phi }_{+}$ is continuous for any $f\in F.$ Then,
the following structural theorem \cite%
{Bere-1,BeKo,BeSh,BoPr,Gold,GoGrPoSh,PrBoGoTa,ReSi-2} holds.

\begin{theorem}
\label{FQth2.1} Assume that the family of operators $\mathcal{A}$ 
satisfies the following conditions:

a) for $ A(f),\;f\in F,$ the closure of the operator $\overline{A(f)}
$ in $ {\Phi }$ coincides with $A(f)$ for any $f\in F,$ that is $\overline{A(f)}=A(f)$ on domain $D_{A(f)}$ in ${\Phi };$

b) the range $A(f) \subset {\Phi }$ for any $f\in F$;

c) for every $|\psi )\in \mathcal{D}$ the mapping $F\ni f\rightarrow
A(f)|\psi )\in {\Phi }_{+}$ is linear and continuous;

d) there exists a strong cyclic vector $|\Omega )\in \bigcap_{f\in
F}D_{A(f)},$ such that the set of all vectors $|\Omega )$ and $%
\prod_{j=1}^{n}A(f_{j})|\Omega ),$ $n\in \mathbb{Z}_{+},$ is total in ${\Phi
}_{+}$ (i.e., their linear hull is dense in ${\Phi }_{+}$).

Then, there exists a probability measure $\mu $ on $(F^{\prime },C_{\sigma
}(F^{\prime }))$, where $F^{\prime }$ is the dual of $F$ and $C_{\sigma
}(F^{\prime })$ is the $\sigma $-algebra generated by cylinder sets in $%
F^{\prime }$ such that, for $\mu$-almost every $\eta \in F^{\prime }$ there
is a generalized joint eigenvector $\omega (\eta )\in {\Phi }_{-}$ of the
family $\mathcal{A},$ corresponding to the joint eigenvalue $\eta \in
F^{\prime },$ that is for any $\varphi \in \mathcal{D}\subset {\Phi }_{+}$
\begin{equation}
(\omega (\eta )|A(f)\varphi )_{\Phi _{-}\times \Phi _{+}}=\eta (f)(\omega
(\eta )|\varphi )_{\Phi _{-}\times \Phi _{+}}  \label{FQeq2.1a}
\end{equation}%
with $\eta (f)\in \mathbb{R}$ denoting here the pairing between $F$ and $%
F^{\prime }.$

The mapping
\begin{equation}
\mathcal{D}\ni |\varphi )\rightarrow (\omega (\eta )|\varphi )_{\Phi
_{-}\times \Phi _{+}}:=\varphi (\eta )\in \mathbb{C}  \label{FQeq2.1b}
\end{equation}%
for any $\eta \in F^{\prime }$ can be continuously extended to a unitary
surjective operator $\mathcal{F}_{\eta }:{\Phi }_{+}\rightarrow
L_{2}^{(\mu )}(F^{\prime };\mathbb{C}),$ where
\begin{equation}
\mathcal{F}_{\eta }\text{ }|\varphi ):=\varphi (\eta )  \label{FQeq2.1c}
\end{equation}%
for any $\eta \in F^{\prime }$ is a generalized Fourier transform,
corresponding to the family $\mathcal{A}.$ Moreover, the image of the
operator $A(f),$ $f\in F^{\prime },$ under the $\mathcal{F}_{\eta }$-mapping is the operator of multiplication by the function $F^{\prime }\ni
\eta \rightarrow \eta (f)\in \mathbb{C}.$
\end{theorem}

Now let us assume that a Hilbert space $\Phi :=\Phi _\text{F}$ possesses the
standard canonical Fock space structure \cite%
{Bere,BeSh,BlPrSa,BoBo,FaYa,GoGrPoSh,PrMy,Takh}, that is
\begin{equation}
 {\Phi }_\text{F}=\oplus _{n\in \mathbb{Z}_{+}} {\Phi }_{(s)}^{\otimes n},
\label{FQeq2.1d}
\end{equation}%
where subspaces ${\Phi }_{(s)}^{\otimes n},$ $n\in \mathbb{Z}_{+}$,
are the symmetrized tensor products of a Hilbert space $\mathcal{H}\simeq
L_{2}^{(s)}(\mathbb{R}^{m};\mathbb{C}).$ If a vector $\varphi :=(\varphi
_{0},\varphi _{1},\ldots,\varphi _{n},\ldots)\in \Phi _\text{F},$ its norm
\begin{equation}
\Vert \varphi \Vert _{\Phi }:=\left( \sum_{n\in \mathbb{Z}_{+}}\Vert \varphi
_{n}\Vert _{n}^{2}\right) ^{1/2},  \label{FQeq2.2}
\end{equation}%
where $\varphi _{n}\in \Phi _{(s)}^{\otimes n}\simeq L_{2}^{(s)}((\mathbb{R}%
^{m})^{\otimes n};\mathbb{C})$ and $\parallel \ldots\parallel _{n}$ is the
corresponding norm in {$\Phi $}$_{(s)}^{\otimes n}$ for all $n\in \mathbb{Z}%
_{+}.$ Note here that concerning the rigging structure (\ref{FQeq2.1}),
there holds the corresponding rigging for the Hilbert spaces {$\Phi $}$%
_{(s)}^{\otimes n},$ $n\in \mathbb{Z}_{+}$, that is
\begin{equation}
\mathcal{D}_{(s)}^{n}\subset \Phi _{(s),+}^{\otimes n}\subset \Phi
_{(s)}^{\otimes n}\subset \Phi _{(s),-}^{\otimes n}  \label{FQeq2.3}
\end{equation}%
with some suitably chosen dense and separable topological spaces of
symmetric functions $\mathcal{D}_{(s)}^{n},$ $n\in \mathbb{Z}_{+}.$
Concerning expansion (\ref{FQeq2.1d}) we obtain by means of projective and
inductive limits \cite{Bere-1,Bere,BeKo,BeSh} the quasi-nucleus rigging of
the Fock space $\Phi $ in the form (\ref{FQeq2.1}).

Consider now any basis vector $|(\alpha )_{n})\in {\Phi }_{(s)}^{\otimes
n},$ $n\in \mathbb{Z}_{+},$ which can be written \cite%
{BeSh,BoBo,Dira,Fock,PrTaBo} in the following canonical Dirac ket-form:
\begin{equation}
|(\alpha )_{n}):=|\alpha _{1},\alpha _{2},\ldots,\alpha _{n}),  \label{FQeq2.4}
\end{equation}%
where, by definition, \
\begin{equation}
|\alpha _{1},\alpha _{2},\ldots,\alpha _{n}):=\frac{1}{\sqrt{n!}}\sum_{\sigma
\in S_{n}}|\alpha _{\sigma (1)})\otimes |\alpha _{\sigma (2)})\ldots|\alpha
_{\sigma (n)})  \label{FQeq2.5}
\end{equation}%
and vectors $ |\alpha _{j})\in $ $\mathcal{H}_{+}\subset $ {$\Phi $}$%
_{(s)}^{\otimes 1}(\mathbb{R}^{m};\mathbb{C})\simeq \mathcal{H},$ $j\in
\mathbb{Z}_{+},$ are bi-orthogonal to each other, that is $(\alpha
_{k}|\alpha _{j})_{\mathcal{H}}=\delta _{k,j}$ for any $k,j\in \mathbb{Z}%
_{+}.$ The corresponding scalar product of base vectors as (\ref{FQeq2.5})
is given as follows:
\begin{align}
((\beta )_{n}|(\alpha )_{n}): &=(\beta _{n},\beta _{n-1},\ldots,\beta _{2},\beta
_{1}|\alpha _{1},\alpha _{2},\ldots,\alpha _{n-1},\alpha _{n}) \nonumber\\
&=\sum_{\sigma \in S_{n}}(\beta _{1}|\alpha _{\sigma (1)})_{\mathcal{H}%
}\ldots(\beta _{n}|\alpha _{\sigma (n)})_{\mathcal{H}}:=\per\{(\beta _{i}|\alpha
_{j})_{\mathcal{H}}\}_{i,j=\overline{1,n}}\,,%
\label{FQeq2.6}
\end{align}%
where ``$\per$'' denotes the permanent of
matrix and $(\cdot |\cdot )_{\mathcal{H}}$ is the corresponding scalar
product in the Hilbert space $\mathcal{H}.$ Based now on the representation (%
\ref{FQeq2.4}), one can define an operator $a^{+}(\alpha ):{\Phi }_{(s)}^{\otimes n}\rightarrow {\Phi }_{(s)}^{\otimes (n+1)}$ for any
$|\alpha \rangle \in \mathcal{H}_{-}$ as follows:
\begin{equation}
a^{+}(\alpha )|\alpha _{1},\alpha _{2},\ldots,\alpha _{n}):=|\alpha ,\alpha
_{1},\alpha _{2},\ldots,\alpha _{n}),  \label{FQeq2.7}
\end{equation}%
which is called the ``\textit{creation}'' operator in the Fock space $\Phi
_\text{F}.$ The adjoint operator $a(\beta ):=(a^{+}(\beta ))^{\ast }:\Phi _{(s)}^{\otimes (n+1)}\rightarrow {\Phi }_{(s)}^{\otimes n}$ with
respect to the Fock space $\Phi _\text{F}$ (\ref{FQeq2.1}) for any $|\beta
\rangle \in H_{-},$ called the ``\textit{annihilation}'' operator, acts as
follows:
\begin{equation}
a(\beta )|\alpha _{1},\alpha _{2},\ldots,\alpha _{n+1}):=\sum_{\sigma \in
S_{n}}(\beta |\alpha _{j})_{\mathcal{H}}|\alpha _{1},\alpha _{2},\ldots,\alpha
_{j-1},\hat{\alpha}_{j},\alpha _{j+1},\ldots,\alpha _{n+1}),  \label{FQeq2.8}
\end{equation}%
where the ``\textit{hat}'' over a vector denotes that it should be omitted from the
sequence.

It is easy to check that the commutator relationship
\begin{equation}
\lbrack a(\alpha ),a^{+}(\beta )]=(\alpha |\beta )_{\mathcal{H}}
\label{FQeq2.9}
\end{equation}%
holds for any vectors $|\alpha )\in \mathcal{H}$ and $|\beta )\in \mathcal{H}%
.$ Expression (\ref{FQeq2.9}), owing to the rigged structure (\ref{FQeq2.1}%
), can be naturally extended to the general case, when vectors $ |\alpha )$
and $|\beta )\in \mathcal{H}_{-},$ conserving its form. In particular,
taking $|\alpha ):=|\alpha (x))=\frac{1}{\sqrt{2\piup }}\re^{\ri\langle \lambda
|x\rangle }\in H_{-}:=L_{2,-}(\mathbb{R}^{m};\mathbb{C})$ for any $x\in
\mathbb{R}^{m},$ one easily gets from (\ref{FQeq2.9}) that
\begin{equation}
\lbrack a(x),a^{+}(y)]=\delta (x-y),  \label{FQeq2.10}
\end{equation}%
where we put, by definition, $\langle\cdot |\cdot \rangle$ the usual scalar product in
the $m$-dimensional Euclidean space $(\mathbb{R}^{m};\langle\cdot |\cdot \rangle),$ $%
a^{+}(x):=a^{+}(\alpha (x))$ and $a(y):=a(\alpha (y))$ for all $x,y\in
\mathbb{R}^{m}$ and denoted by $\delta (\cdot )$ the classical Dirac
delta-function.

The above construction makes it possible to observe easily that there exists
the unique vacuum vector $|0)\in \mathcal{\Phi }_{(s)}^{\otimes n}|_{n=0},$
 such that for any $x\in \mathbb{R}^{m}$
\begin{equation}
a(x)|0)=0,  \label{FQeq2.11}
\end{equation}%
and the set of vectors
\begin{equation}
\left( \prod_{j=1}^{n}a^{+}(x_{j})\right) |0)\in {\Phi }_{(s)}^{\otimes n}
\label{FQeq2.12}
\end{equation}%
is total in ${\Phi }_{(s)}^{\otimes n},$ that is their linear integral hull
over the functional spaces ${\Phi }_{(s)}^{\otimes n}$ is dense in the
Hilbert space ${\Phi }_{(s)}^{\otimes n}$ for every $n\in \mathbb{Z}_{+}.$
This means that for any vector $\varphi \in \Phi _\text{F}$, the following
canonical representation
\begin{equation}
\varphi =\sum\limits_{n\in \mathbb{Z}_{+}}^{\oplus }\frac{1}{\sqrt{n!}}%
\int_{(\mathbb{R}^{m})^{n}}\varphi
_{n}(x_{1},\ldots,x_{n})a^{+}(x_{1})a^{+}(x_{2})\ldots a^{+}(x_{n})|0)
\label{FQeq2.13}
\end{equation}%
holds with the Fourier type coefficients $\varphi _{n}\in \Phi
_{(s)}^{\otimes n}$ for all $n\in \mathbb{Z}_{+}$ with $\Phi _{(s)}^{\otimes
1}\simeq \mathcal{H}={L}_{2}(\mathbb{R}^{m};\mathbb{C}).$ The latter is
naturally endowed with the Gelfand type quasi-nucleus rigging, dual to
\begin{equation}
\mathcal{H}_{+}\subset \mathcal{H}\subset \mathcal{H}_{-}\,,  \label{FQeq2.14}
\end{equation}%
making it possible to construct a quasi-nucleus rigging of the dual Fock
space $ {\Phi }_\text{F}:=\oplus _{n\in \mathbb{Z}_{+}} {\Phi }_{(s)}^{\otimes
n}.$ Thereby, the chain (\ref{FQeq2.14}) generates the dual Fock space
quasi-nucleolus rigging
\begin{equation}
\mathcal{D}\subset {\Phi }_{F,+}\subset \Phi _\text{F}\subset \Phi
_{F,-}\subset \mathcal{D}^{\prime }  \label{FQeq2.15}
\end{equation}%
with respect to the Fock space $ \Phi _\text{F},$ easily following from (\ref%
{FQeq2.1}) and (\ref{FQeq2.14}).

Construct now the following self-adjoint operator $ \rho (x):\Phi
_\text{F}\rightarrow \Phi _\text{F}$  as
\begin{equation}
\rho (x):=a^{+}(x)a(x),  \label{FQeq2.16}
\end{equation}%
called the density operator at a point $x\in \mathbb{R}^{m},$ satisfying the
commutation properties:
\begin{equation}
\begin{array}{c}
\lbrack \rho (x),\rho (y)]=0, \\[5pt]
\lbrack \rho (x),a(y)]=-a(y)\delta (x-y), \\[5pt]
\lbrack \rho (x),a^{+}(y)]=a^{+}(y)\delta (x-y)%
\end{array}
\label{FQeq2.17}
\end{equation}%
for any $x,y\in \mathbb{R}^{m}.$

Now, if we construct the following self-adjoint family $\mathcal{R}:=\big\{
\int_{\mathbb{R}^{m}}\rho (x)f(x)\rd x:f\in F\big\} $ of linear operators in
the Hilbert space $\Phi ,$ where $F$ $:=\mathcal{S}(\mathbb{R}^{m};\mathbb{R%
})$ is the Schwartz functional space dense in $H,$ one can derive, making
use of theorem \ref{FQth2.1}, that there exists the generalized Fourier
transform (\ref{FQeq2.1c}), such that
\begin{equation}
{\Phi }=L_{2}^{(\mu )}(F^{\prime };\mathbb{C})\simeq \int_{F^{\prime
}}^{\oplus }\Phi _{\eta }\rd\mu (\eta )  \label{FQeq2.17a}
\end{equation}%
for some Hilbert space sets $\Phi _{\eta },$ $\eta \in F^{\prime },$ and a
suitable measure $\mu $ on $F^{\prime },$ with respect to which the
corresponding joint eigenvector $\omega (\eta )\in \Phi _{-}$ for any $\eta
\in F^{\prime }$ generates the Fourier transformed family $ \left\{ \eta
(f)\in \mathbb{R}:f\in F\right\}.$ Moreover, if $\dim \Phi _{\eta }=1$
for all $ \eta \in F^{\prime },$ the Fourier transformed eigenvector $%
\omega (\eta ):=\Omega (\eta )=1$ for all $\eta \in F^{\prime }.$

Now, we  consider the family of self-adjoint operators $\rho (f):\Phi
_{\eta }\rightarrow \Phi _{\eta },f\in F,$ as generating a unitary family $%
\mathcal{U}:=\left\{ \mathrm{U}(f):f\in F\right\} ,$ where the operator
\begin{equation}
\mathrm{U}(f):=\exp [\ri\rho (f)]  \label{FQeq2.18}
\end{equation}%
is unitary, satisfying the abelian commutation condition
\begin{equation}
\mathrm{U}(f_{1})\mathrm{U}(f_{2})=\mathrm{U}(f_{1}+f_{2})  \label{FQeq2.19}
\end{equation}%
for any $f_{1},f_{2}\in F.$ Since, in general, the unitary family $%
\mathcal{U}$ is defined on the Hilbert space $\Phi _{\eta },$ the important
problem of describing its cyclic representation spaces arises, within which
the factorization
\begin{equation}
\rho (f)=\int_{\mathbb{R}^{m}}a^{+}(x)a(x)f(x)\rd x  \label{FQeq2.20}
\end{equation}%
jointly with relationships (\ref{FQeq2.17}) hold for any $f\in F.$ This
problem can be treated using mathematical tools devised both within the
representation theory of $\mathbb{C}^{\ast}$-algebras \cite%
{BeSh,Dira,Gold,GoGrPoSh} and the Gelfand-Vilenkin \cite{GeVi} approach.
Below we will describe the main features of the Gelfand-Vilenkin formalism,
being much more suitable for the task, providing a reasonably unified
framework of constructing the corresponding representations. The next
definitions will be used in our construction.

\begin{definition}
\label{FQdef2.1} Let $F$ be a locally convex topological vector space, $%
F_{0}\subset F$ be a finite dimensional subspace of $F.$ Let $%
F^{0}\subseteq F^{\prime }$ be defined by
\begin{equation}
F^{0}:=\left\{ \sigma \in F^{\prime }:\sigma |_{F_{0}}=0\right\} ,
\label{FQeq2.21}
\end{equation}%
and called the annihilator of $F_{0}$.
\end{definition}

The quotient space $F^{\prime \, 0}:=F^{\prime }/F^{0}$ may be, evidently,
identified with $F_{0}^{\prime }\subset F^{\prime },$ the adjoint space of $%
F_{0}.$

\begin{definition}
\label{FQdef2.2} Let $Q\subseteq F^{\prime\, 0};$ then, the subset
\begin{equation}
X_{F^{0}}^{(Q)}:=\left\{ \sigma \in F^{\prime }:\sigma +F^{0}\subset
Q\right\}  \label{FQeq2.22}
\end{equation}%
is called the cylinder set with the base $Q$ and the generating subspace $%
F^{0}.$
\end{definition}

\begin{definition}
\label{FQdef2.3} Let $n=\dim F_{0}=\dim F_{0}^{\prime }=\dim F^{\prime \,0}.$
One says that a cylinder set $X^{(Q)}$ has Borel base, if $Q$ is a Borel
set, when regarded as a subset of $\mathbb{R}^{m}.$
\end{definition}

The family of cylinder sets with Borel base forms an algebra of sets, which
is a key stone for defining measurable sets in $F^{\prime}$ and the corresponding
measures on $F^{\prime }.$

\begin{definition}
\label{FQdef2.4} The measurable sets in $F^{\prime }$ are the elements of
the $\sigma $-algebra generated by the cylinder sets with Borel base.
\end{definition}

\begin{definition}
\label{FQdef2.5} A cylindrical measure in $F^{\prime }$ is a non-negative $%
\sigma $-pre-additive function $\mu $ defined on the algebra of cylinder
sets with Borel base and satisfying the conditions $0\leqslant \mu (X)\leqslant 1$ for
any $X,$ $\mu (F^{\prime })=1$ and $\mu ( \coprod_{j\in \mathbb{Z}%
_{+}}X_{j}) =\sum_{j\in \mathbb{Z}_{+}}\mu (X_{j}),$ if all sets $%
X_{j}\subset F^{\prime },$ $j\in \mathbb{Z}_{+},$ have a common generating
subspace $F_{0}\subset F$.
\end{definition}

\begin{definition}
\label{FQdef2.6} A cylindrical measure $\mu $ satisfies the commutativity
condition if and only if for any bounded continuous function $\alpha :%
\mathbb{R}^{n}\rightarrow \mathbb{R}$ of $n\in \mathbb{Z}_{+}$ real
variables the function
\begin{equation}
\alpha \lbrack f_{1},f_{2},\ldots ,f_{n}]:=\int_{F^{\prime }}\alpha (\eta
(f_{1}),\eta (f_{2}),\ldots ,\eta (f_{n}))\rd\mu (\eta )  \label{FQeq2.23}
\end{equation}%
is sequentially continuous in $f_{j}\in F,$ $j=\overline{1,m}.$
\end{definition}

\begin{remark}
It is known \cite{GeVi,Gold,BeSh} that in countably normalized spaces, the
properties of sequential and ordinary continuity are equivalent.
\end{remark}

\begin{definition}
\label{FQdef2.7} A cylindrical measure $\mu $ is countably additive if and
only if for any cylinder set $X=\coprod_{j\in \mathbb{Z}_{+}}X_{j},$ which
is the union of countable numerous mutually disjoint cylinder sets $%
X_{j}\subset F^{\prime },j\in \mathbb{Z}_{+},$ $\mu (X)=\sum_{j\in \mathbb{Z}%
_{+}}\mu (X_{j}).$
\end{definition}

The next two standard propositions \cite{AlDaKoLy,AlKoSt,BeSh,GeVi,Gold},
characterizing extensions of the measure $\mu $ on $ X=\coprod_{j\in
\mathbb{Z}_{+}}X_{j},$ hold.

\begin{proposition}
\label{FQpr2.8} A countably additive cylindrical measure $\mu $ can be
extended to a countably additive measure on the $\sigma $-algebra, 
generated by the cylinder sets with Borel base. Such a measure will also be
called a cylindrical measure.
\end{proposition}

\begin{proposition}
\label{FQpr2.9} Let $F$ be a nuclear space. Then, any cylindrical measure $%
\mu $ on $F^{\prime },$ satisfying the continuity condition, is countably
additive.
\end{proposition}

Concerning the Fourier transform of a cylindrical measure $\mu $ in $%
F^{\prime },$ we will use the following natural definitions.

\begin{definition}
\label{FQdef2.10} Let $\mu $ be a cylindrical measure in $F^{\prime }$. The
Fourier transform of $\mu $ is the nonlinear functional
\begin{equation}
\mathcal{L}(f):=\int_{F^{\prime }}\exp [\ri\eta (f)]\rd\mu (\eta ),
\label{FQeq2.24}
\end{equation}%
coinciding with the characteristic functional of the measure $\mu .$
\end{definition}

\begin{definition}
\label{FQdef2.11} The nonlinear functional $\mathcal{L}:F\rightarrow
\mathbb{C}$ on $F,$ defined by (\ref{FQeq2.24}), is called positive
definite, if and only if for all $f_{j}\in F$ and $\lambda _{j}\in \mathbb{C}%
,$ $j=\overline{1,n},$ the condition
\begin{equation}
\sum_{j,k=1}^{n}\bar{\lambda}_{j}\mathcal{L}(f_{k}-f_{j})\lambda _{k}\geqslant 0
\label{FQeq2.25}
\end{equation}%
holds for any $n\in \mathbb{Z}_{+}.$
\end{definition}

The following important proposition, owing to Gelfand and Vilenkin \cite%
{GeVi,Gold}, Araki \cite{Arak} and Goldin \cite{Gold,Gold-1}, holds.

\begin{proposition}
\label{FQpr2.12} The functional $\mathcal{L}:F\rightarrow \mathbb{C}$ on
$F,$ defined by (\ref{FQeq2.24}), is the Fourier transform of a cylindrical
measure on $F^{\prime }$ if and only if it is positive definite,
sequentially continuous and satisfying the condition $\mathcal{L}(0)=1$%
.\bigskip\
\end{proposition}

\subsection{The unitary family and generating functional equations}

Suppose now that we have a continuous unitary representation of the unitary
family $\mathcal{U}$ on a suitable Hilbert space $\Phi $ with a cyclic
vector $|\Omega )\in \Phi .$ Then, we can put
\begin{equation}
\mathcal{L}(f):=(\Omega |\mathrm{U}(f)|\Omega )  \label{FQeq2.26}
\end{equation}%
for any $f\in F:=\mathcal{S},$ being the Schwartz space on $\mathbb{R}^{m},$
and observe that functional (\ref{FQeq2.26}) is continuous on $F$ owing to
the continuity of the representation. Therefore, this functional is the
generalized Fourier transform of a cylindrical measure $\mu $ on $F^{\prime
}$:
\begin{equation}
(\Omega |\mathrm{U}(f)|\Omega )=\int_{\mathcal{S}^{\prime }}\exp [\ri\eta
(f)]\rd\mu (\eta ).  \label{FQeq2.27}
\end{equation}%
From the spectral point of view, based on theorem \ref{FQth2.1}, there is an
isomorphism between the Hilbert spaces $\Phi $ and $L_{2}^{(\mu )}(F;\mathbb{%
C}),$ defined by $|\Omega )\rightarrow \Omega (\eta )=1$ and $\mathrm{U}%
(f)|\Omega )\rightarrow \exp [\ri\eta (f)]$ and next extended by linearity
upon the whole Hilbert space $\Phi $.

In the non-cyclic case, there exists a finite or countably
infinite family of measures $\left\{ \mu _{k}:k\in \mathbb{Z}_{+}\right\} $
on $F^{\prime },$ with ${\Phi }\simeq \oplus _{k\in \mathbb{Z}%
_{+}}L_{2}^{(\mu _{k})}(F^{\prime };\mathbb{C})$ and the unitary operator $%
\mathrm{U}(f):{\Phi }\rightarrow {\Phi }$ for any $f\in \mathcal{S}%
^{\prime }$ corresponds in all $L_{2}^{(\mu _{k})}(F^{\prime };\mathbb{C}),$
$k\in \mathbb{Z}_{+},$ to a multiplication operator on the exponent function
$\exp [\ri\eta (f)].$ This means that there exists a single cylindrical
measure $\mu $ on $F^{\prime }$ and a $\mu$-measurable field of Hilbert
spaces ${\Phi }_{\eta }$ on $F^{\prime },$ such that
\begin{equation}
{\Phi }\simeq \int_{F^{\prime }}^{\oplus }{\Phi }_{\eta }\rd\mu (\eta ),
\label{FQeq2.28}
\end{equation}%
with $\mathrm{U}(f):{\Phi }\rightarrow {\Phi },$ corresponding \cite%
{GeVi} to the operator of multiplication by $\exp [\ri\eta (f)]$ for any $f\in
F$ and $\eta \in F^{\prime }.$ Thereby, having constructed the nonlinear
functional (\ref{FQeq2.24}) in an exact analytical form, one can retrieve
the representation of the unitary family $\mathcal{U}$ in the corresponding
Hilbert space ${\Phi },$ making use of the suitable factorization (\ref%
{FQeq2.20}) as follows: ${\Phi }=\oplus _{n\in \mathbb{Z}_{+}}{\Phi }_{n},$
where
\begin{equation}
{\Phi }_{n}=\prod_{j=\overline{1,n}}\rho (x_{j})|\Omega ),  \label{FQeq2.29}
\end{equation}%
for all $n\in \mathbb{Z}_{+}.$

The cyclic vector $|\Omega )\in {\Phi }$ can be, in particular, obtained as
the ground state vector of some unbounded self-adjoint positive definite
Hamilton operator $\mathrm{\hat{H}}:{\Phi }\rightarrow {\Phi },$
commuting with the self-adjoint non-negative particle number operator
\begin{equation}
\mathrm{N}:=\int_{\mathbb{R}^{m}}\rd x\rho (x),  \label{FQeq2.30}
\end{equation}%
that is $[\mathrm{\hat{H}},\mathrm{N}]=0$. Moreover, the conditions
\begin{equation}
\mathrm{\hat{H}}|\Omega )=0  \label{FQeq2.31}
\end{equation}%
and
\begin{equation}
\inf_{\varphi \in D_{\mathrm{H}}}(\varphi |\mathrm{\hat{H}}|\varphi
)=(\Omega |\mathrm{\hat{H}}|\Omega )=0  \label{FQeq2.32}
\end{equation}%
hold for the operator $\mathrm{\hat{H}}:{\Phi }\rightarrow {\Phi },$
where $D_{\mathrm{H}}$ denotes its domain of definition, dense in $\Phi .$
To find the functional (\ref{FQeq2.26}), which is called the generating
Bogolubov type functional for moment distribution functions
\begin{equation}
F_{n}(x_{1},x_{2},\ldots ,x_{n}):=(\Omega |:\rho (x_{1})\rho (x_{2})\ldots \rho
(x_{n}):|\Omega ),  \label{FQeq2.33}
\end{equation}%
where $x_{j}\in \mathbb{R}^{m},$ $j=\overline{1,n},$ and the \textit{normal
ordering operation} $:\cdot :$ is defined \cite%
{Bere,BoBo,BoPr,Gold,GoMeSh-1,GoMeSh-2} as
\begin{equation}
:\rho (x_{1})\rho (x_{2})\ldots \rho (x_{n}):\;=\prod_{j=1}^{n}\left( \rho
(x_{j})-\sum_{k=1}^{j-1}\delta (x_{j}-x_{k})\right) ,  \label{FQeq2.34}
\end{equation}%
it is convenient first to choose the Hamilton operator $\mathrm{H:}$ $\Phi
_\text{F}\rightarrow \Phi _\text{F} $ in the following secondly quantized \cite%
{BoBo,Gold,GoGrPoSh} representation
\begin{equation}
\mathrm{H:}=\frac{1}{2}\int_{\mathbb{R}^{m}}\left\langle \nabla
_{x}a^{+}(x)|\nabla _{x}a(x)\right\rangle \rd x+\mathrm{V}(\rho ),
\label{FQeq2.35}
\end{equation}%
in the related Fock space $\Phi _\text{F},$ where the sign ``$\nabla _{x}$''
means the usual gradient operation with respect to $x\in \mathbb{R}^{m}$ in
the Euclidean space $\mathbb{ E}^{m}\simeq (\mathbb{R}^{m};\left\langle
\cdot |\cdot \right\rangle ).$ If the energy spectrum density of the
Hamiltonian operator (\ref{FQeq2.35}) in the Fock space {${\Phi }_\text{F} $}
is bounded from below, in works done by Goldin, Grodnik, Menikov,
Powers and Sharp \cite{Meni,Gold,GoGrPoSh} there was stated that
this Hamiltonian, modulo the ground state energy eigenvalue, can be
algebraically represented on a suitably constructed \textit{current algebra
symmetry representation Hilbert space} ${\Phi },$ as the positive definite
gauge type factorized operator
\begin{equation}
\mathrm{\hat{H}}=\frac{1}{2}\int_{\mathbb{R}^{m}}\left\langle (\mathrm{K}%
^{+}(x)-\mathrm{A}(x;\rho ))|\rho ^{-1}(x)(\mathrm{K}(x)-\mathrm{A}(x;\rho
))\right\rangle \rd x,  \label{FQeq2.35b}
\end{equation}%
satisfying conditions (\ref{FQeq2.31}) and (\ref{FQeq2.32}), where $\mathrm{A%
}(x;\rho ):\Phi \rightarrow \Phi ^{m},$ $ x\in \mathbb{R}^{m},$ is some
specially constructed linear self-adjoint operator, satisfying the condition
\begin{equation}
\mathrm{K}(x)|\Omega )=\mathrm{A}(x;\rho )|\Omega )  \label{FQeq2.38}
\end{equation}%
for any $x\in \mathbb{R}^{m}$ and the ground state $|\Omega )\in {\Phi }%
, $ corresponding to a chosen potential operators $\mathrm{V}(\rho ):{\Phi
}\rightarrow {\Phi }.$ The singular structure of the operator (\ref%
{FQeq2.35b}) was earlier analyzed in detail in \cite{GoSh-3} where, in part,
there was shown its wellposedness.

The ``potential'' operator $\mathrm{V}(\rho ):{\Phi }%
\rightarrow {\Phi }$ is, in general, a polynomial (or analytical)
functional of the density operator $\rho (x):{\Phi }\rightarrow {\Phi }$
for any $x\in \mathbb{R}^{m},$ and the operator $\mathrm{K}(x):{\Phi }%
\rightarrow {\Phi }^{m}$ is defined as
\begin{equation}
\mathrm{K}(x):=\nabla _{x}\rho (x)/2+\ri J(x),  \label{FQeq2.36}
\end{equation}%
where the self-adjoint ``current''  operator $J(x):{\Phi }\rightarrow
{\Phi }^{m}$ can be naturally defined (but non-uniquely) from the continuity
equality
\begin{equation}
\partial \rho /\partial t=-\ri[\mathrm{H},\rho (x)]=-\langle\nabla |J(x)\rangle,
\label{FQeq2.37}
\end{equation}%
holding for all $x\in \mathbb{R}^{m}.$ Such an operator $J(x):{\Phi }%
\rightarrow {\Phi }^{m},$ $x\in \mathbb{R}^{m},$ can exist owing to the
commutation condition $[\mathrm{\hat{H}},\mathrm{\hat{N}}]=0,$ giving rise
to the continuity relationship (\ref{FQeq2.37}), if we additionally take
into account that supports $\mathrm{supp}$ $\rho $ of the density operator $%
\rho (x):{\Phi }\rightarrow {\Phi },$ $x\in \mathbb{R}^{m},$ can be
chosen arbitrarily owing to the independence of (\ref{FQeq2.37}) on the
potential operator $V(\rho ):{\Phi }\rightarrow {\Phi },$ but its strict
dependence on the corresponding representation (\ref{FQeq2.28}).

\begin{remark}
Here we mention that the operator $\mathrm{K}(x):{\Phi }\rightarrow {%
\Phi }^{m}, x\in \mathbb{R}^{m},$ defined by (\ref{FQeq2.36}), relates to
that from the work \cite{Gold,GoGrPoSh,Meni} via scaling $\mathrm{K}%
(x)\rightarrow \mathrm{K}(x)/2,$ $x\in \mathbb{R}^{m}.$
\end{remark}

In particular, based on the Fock space $\Phi _\text{F},$ defined by (\ref{FQeq2.1}%
) and generated by the creation-annihilation operators (\ref{FQeq2.7}) and (%
\ref{FQeq2.8}), the current operator $ J(x):{\Phi }_\text{F}\rightarrow {%
\Phi }_\text{F}^{{m}},$ $x\in \mathbb{R}^{m},$ can be easily constructed as
follows:
\begin{equation}
J(x)=\frac{1}{2\ri}[a^{+}(x)\text{ }\nabla _{x}a(x)-\nabla _{x}a^{+}(x)\text{ }%
a(x)],  \label{FQeq2.37a}
\end{equation}%
satisfying jointly with the density operator $\rho (x):{\Phi }%
_\text{F}\rightarrow {\Phi }_\text{F},$ $x\in \mathbb{R}^{m},$ defined by (\ref%
{FQeq2.16}), the following quantum current Lie algebra symmetry \cite%
{Aref,BoPr,Gold,GoGrPoSh,GoMeSh-1,GoMeSh-2,PrMy} relationships:%
\begin{align}
\lbrack J(g_{1}),J(g_{2})] &=\ri J([g_{1,}g_{2}]),\text{ }[\rho
(f_{1}),\rho (f_{2})]=0,    \notag\\
\lbrack J(g_{1}),\rho (f_{1})] &=\ri \rho (\left\langle g_{1}|\nabla
f_{1}\right\rangle),\text{ }  \label{FQeq2.37b}
\end{align}%
holding for all $f_{1},f_{1}\in F$ and $g_{1},g_{2}\in F^{m},$ where we
put, by definition,
\begin{equation}
\lbrack g_{1,}g_{2}]:=\left\langle g_{1}|\nabla \right\rangle g_{2}-\left\langle g_{2}|\nabla \right\rangle g_{1},  \label{FQeq2.37c}
\end{equation}%
being the usual commutator of vector fields $\left\langle g_{1}|\nabla
\right\rangle _{\ }$and $\left\langle g_{2}|\nabla \right\rangle _{\ }$on
the configuration space $\mathbb{R}^{m}.$ It is easy to observe that the
current algebra (\ref{FQeq2.37b}) is the Lie algebra $\mathcal{G},$
corresponding to the Banach Lie group $G=\Diff(\mathbb{R}^{m})\ltimes F,$\
the semidirect product of the Banach-Lie group of diffeomorphisms  $\Diff (\mathbb{R}^{m})$ of the $m$-dimensional space $\mathbb{R}^{m}$ 
and the abelian group $F$ subject to the multiplicative operation Banach
group of smooth functions $F.$

\begin{remark}
The self-adjointness of the operator $\mathrm{A}(g;\rho ):{\Phi }%
\rightarrow {\Phi }^{m},$ $g\in F,$ can be stated following schemes from
works \cite{Aref,BoPr,GoGrPoSh} under the additional existence of such a
linear anti-unitary mapping $\mathrm{T}:{\Phi }\rightarrow {\Phi }$ that
the following invariance conditions hold:
\begin{equation}
\mathrm{T}\rho (x)\mathrm{T}^{-1}=\rho (x),\qquad \mathrm{T}J(x)\mathrm{T%
}^{-1}=-J(x),\qquad \mathrm{T}|\Omega )=|\Omega )  \label{FQeq2.39}
\end{equation}%
for any $x\in \mathbb{R}^{m}.$ Thereby, owing to conditions (\ref{FQeq2.39}%
), the following equalities
\begin{equation}
\mathrm{K}(x)|\Omega )=\mathrm{A}(x;\rho )|\Omega )  \label{FQeq2.40}
\end{equation}%
hold for any $x\in \mathbb{R}^{m},$ giving rise to the self-adjointness of
the operator $\mathrm{A}(g;\rho ):{\Phi }\rightarrow {\Phi }^{m},$ $g\in
F.$
\end{remark}

It is easy to observe that the time-reversal condition (\ref{FQeq2.39})
imposes the real value relationship for the ground state $\ \Omega _{N}=%
\overline{\Omega }_{N}\in {\Phi }_{N}$ ${\simeq L}_{2}^{(s)}(\mathbb{R}%
^{m\times N};\mathbb{C})$ of the canonically represented $N$-particle
Hamiltonian \textrm{H}$_{N}:{\Phi }_{N}{\rightarrow \Phi }_{N}{\ }$for
arbitrary $N\in \mathbb{Z}_{+}.$ Moreover, taking into account$\ $the
relationship\ (\ref{FQeq2.40}), one can easily observe that on the invariant
subspace ${\Phi }_{N}\subset \Phi $ the operator $K(x):{\Phi }%
\rightarrow {\Phi }^{m}$ is representable as
\begin{equation}
K_{N}(x)=\sum_{j=\overline{1,N}}\delta (x-x_{j})\frac{\partial }{\partial
x_{j}}\,,  \label{FQeq2.41b}
\end{equation}%
entailing the following expression for the related operator $A_{N}(x;\rho ):{%
\Phi }_{N}{\rightarrow \Phi }_{N}^{m}{\ }$\ on the subspace ${\Phi }%
_{N}\subset \Phi :$ ${\ }$
\begin{equation}
\ A_{N}(x;\rho )=\sum_{j=\overline{1,N}}\delta (x-x_{j})\nabla _{x_{j}}\ln
|\Omega _{N}(x_{1},x_{2},\ldots ,x_{N})|.  \label{FQeq2.41c}
\end{equation}%
The latter makes it possible to derive its secondly quantized \cite%
{BeSh,BoBo} expression as
\begin{eqnarray}
\mathrm{A}(x;\rho ) =\int_{\mathbb{R}^{m\times
N}}\rd x_{2}\rd x_{3}\ldots \rd x_{N}:\rho (x)\rho (x_{2})\rho (x_{3})\ldots 
 \rho (x_{N}) :\nabla _{x}\ln |\Omega _{N}(x,x_{2},\ldots ,x_{N})|,  \label{FQeq2.42}
\end{eqnarray}%
which holds for any $x\in \mathbb{R}^{m}$ and arbitrary $N\in \mathbb{Z}%
_{+}. $ Being interested in the infinite particle case when $N\rightarrow
\infty ,$ the expression (\ref{FQeq2.42}) can be naturally decomposed \cite%
{Feen,MeSh} as
\begin{align}
&\mathrm{A}(x;\rho ):= \rho (x)\nabla \frac{\delta }{\delta \rho (x)}%
\mathrm{W}\left( \rho \right) \nonumber\\
&=\sum_{n\mathbb{\in Z}_{+}}\frac{1}{n!}\int_{\mathbb{R}^{m\times
n}}\rd y_{1}\rd y_{2}\ldots \rd y_{n}:\rho (x)\rho (y_{1})\rho (y)\rho (y_{3})
\ldots \rho (y_{n}):\nabla _{x}W_{n+1}(x;y_{1},y_{2},\ldots ,y_{n}),%
\label{FQeq2.43}
\end{align}%
where the corresponding real-valued coefficients $W_{n}$ $\in H_{2\ \
}^{(1)}(\mathbb{R}^{m\times n};\mathbb{R})$ should be such functions that
the series (\ref{FQeq2.43}) were convergent in a suitably chosen
representation Fock space ${\Phi }_\text{F}{,}$ for which the resulting ground
state $\lim_{N\rightarrow \infty }\Omega _{N}\simeq $ $|\Omega )\in {\Phi }$
is necessarily cyclic in $\Phi $ and normalized.

Based now on the above construction one easily deduces from expression (\ref%
{FQeq2.36}) that the generating Bogolubov type functional (\ref{FQeq2.26})
obeys, for all $x\in \mathbb{R}^{m}$, the following functional-differential
equation:
\begin{equation}
\left[ \nabla _{x}-\ri\nabla _{x}f\right] \frac{1}{\ri}\frac{\delta \mathcal{L}%
(f)}{\delta f(x)}=\mathrm{A}\left( x;\frac{1}{\ri}\frac{\delta }{\delta f}%
\right) \mathcal{L}(f),  \label{FQeq2.41}
\end{equation}%
whose solutions should satisfy \cite{GoSh-2} the Fourier transform
representation (\ref{FQeq2.27}), and which were, in part, studied in \cite%
{GoSh-2}. In particular, a wide class of special so-called Poissonian white
noise type solutions to the functional-differential equation (\ref{FQeq2.41}%
) was obtained in \cite{AlDaKoLy,Bere-1,Bere-2,BeKo,BoPr,GoGrPoSh} by means
of functional-operator methods in the following generalized form:
\begin{equation}
\mathcal{L}(f)=\exp \left[ \int_{\mathbb{R}^{m}}\mathrm{W}\left( \frac{1}{\ri}%
\frac{\delta }{\delta f}\right) \rd x\right] \exp \left( \bar{\rho}\int_{%
\mathbb{R}^{m}}\{\exp [\ri f(x)]-1\}\rd x\right) ,  \label{FQeq2.41a}
\end{equation}%
where $\bar{\rho}=(\Omega |\rho |\Omega )\in \mathbb{R}_{+}$ is a suitable
Poisson process parameter and the operator $\mathrm{A}\left( x;\rho \right)
:\Phi \rightarrow \Phi ^{m},x\in \mathbb{R}^{m},$ resulting from the
expression (\ref{FQeq2.43}) for some scalar operator $\mathrm{W}\left(
\rho \right) :\Phi \rightarrow \Phi .$

\begin{remark}
It is worth to remark here that solutions to equation (\ref{FQeq2.41})
realize the suitable physically motivated representations of the abelian
Banach subgroup $F$ of the Banach group $G=\Diff(\mathbb{R}^{m})\ltimes F,$
mentioned above. In the general case of this Banach group $G$, one can
also construct \cite{GoMeSh-1,GoGrPoSh,PrBoGoTa} a generalized Bogolubov
type functional equation, whose solutions give rise to suitable physically
motivated representations of the corresponding current Lie algebra $\mathcal{%
G}.$
\end{remark}

Recalling now the Hamiltonian operator representation (\ref{FQeq2.35b}),
one can readily deduce that the following weak representation Hilbert space $%
{\Phi }$ is a weak relationship
\begin{equation}
\left( \left\langle \mathrm{A}|\rho ^{-1}\mathrm{A}\right\rangle
-\left\langle \mathrm{K}^{\ast }|\rho ^{-1}\mathrm{A}\right\rangle
-\left\langle \mathrm{A}|\rho ^{-1}\mathrm{K}\right\rangle \right) \big/2-
\mathrm{V}(\rho )= \epsilon _{0}\,,  \label{FQeq2.44}
\end{equation}%
where $\epsilon _{0}\in \mathbb{R}$ is the corresponding ground state energy
density value. Thus, the main analytical problem is now reduced to
constructing the expansion (\ref{FQeq2.43}) corresponding to a suitable
cyclic representation Hilbert space ${\Phi }$ of the quantum current algebra
 (\ref{FQeq2.37b}), compatible with the Hamiltonian operator structure~(%
\ref{FQeq2.35}).

\subsection{The Hamiltonian operator reconstruction within the cyclic
current algebra representation}

We assume that we are given a Banach current group $G=\Diff(\mathbb{R}%
^{m})\ltimes F$ cyclic representation in a Hilbert space $\Phi $ with
respect to $F$ with a cyclic vector $|\Omega )\in \Phi _{+}\subset \Phi .$
Based on the well-known Araki reconstruction theorem \cite{Arak,GoGrPoSh} \
for the canonical Weyl commutation relations, from (\ref{FQeq2.37}) we can first readily obtain
\begin{equation}
\lbrack \mathrm{\hat{H},U}(f)]=J(\nabla f)\mathrm{U}(f)-\frac{1}{2}\rho
(\left\langle \nabla f_{1}|\nabla f_{2}\right\rangle )\mathrm{U}(f),
\label{FQeq4.1}
\end{equation}%
where $\mathrm{U}(f)=\exp [\ri\rho (f)],f\in F,$ is an element of the unitary
family $\mathcal{U}.$ The expression (\ref{FQeq4.1}) makes it possible to
calculate the bilinear form
\begin{equation}
(\mathrm{U}(f_{1})\Omega |\mathrm{\hat{H}|U}(f_{2})\Omega )=(\mathrm{U}%
(f_{1})\Omega |J(\nabla f_{1})|\mathrm{U}(f_{2})\Omega )
-\frac{1}{2}(\mathrm{U}(f_{1})\Omega |\rho (\left\langle \nabla f_{1}|\nabla
f_{2}\right\rangle )|\mathrm{U}(f_{2})\Omega )%
\label{FQeq4.2}
\end{equation}%
for any $f_{1},f_{2}\in F.$ Taking into account the symmetry properties (%
\ref{FQeq2.39}), we finally deduce from (\ref{FQeq4.2}) that for arbitrary
functions $f_{1},f_{2}\in F$
\begin{equation}
(\mathrm{U}(f_{1})\Omega |\mathrm{\hat{H}|U}(f_{2})|\Omega )=\frac{1}{2}(\mathrm{U}%
(f_{1})\Omega |\rho (\left\langle \nabla f_{1}|\nabla f_{2}\right\rangle )|%
\mathrm{U}(f_{2})\Omega ).  \label{FQeq4.3}
\end{equation}%
The following standard reasonings make it possible to state that the
bilinear symmetric form (\ref{FQeq4.3}) determines on $\Phi $ a
self-adjoint non-negative definite Hamiltonian operator $\mathrm{\hat{H}}%
:\Phi \rightarrow \Phi ,$ densely defined on the domain $D_{\mathrm{H}}:=
\rm{span}_{f\in F}\{\exp [\ri\rho (f)]|\Omega )\in \Phi \}.$ Really,
for any set of functions $f_{j}\in F,j=\overline{1,n},$ the following
inequalities
\begin{equation}
\sum\limits_{j,k=\overline{1,n}}\bar{s}_{j}s_{k}\left\langle \nabla
f_{j}|\nabla f_{k}\right\rangle \geqslant 0, \qquad 
\sum\limits_{j,k=%
\overline{1,n}}\bar{s}_{j}s_{k}(\mathrm{U}(f_{j})\Omega |\rho (x)|\mathrm{U}%
(f_{k})\Omega )\geqslant 0  \label{FQeq4.4}
\end{equation}%
hold for any complex numbers $s_{j}\in \mathbb{C},j=\overline{1,n},$ and
arbitrary $n\in \mathbb{N}.$ Since for any non-negative definite complex
matrices $A,B\in \End\mathbb{R}^{n}$, the matrix $C:=\{A_{jk}B_{jk}:j,k=%
\overline{1,n}\}\in \End \mathbb{R}^{n}$ proves to be non-negative definite
\cite{Arak,BeBe} too, one ensues that the bilinear form (\ref{FQeq4.4}) is
also non-negative definite. Then, as follows from the classical Friedrichs'
 theorem \cite{Frie1,Frie2,Frie3,Kato,ReSi-1,ReSi-2}, there exists a self-adjoint
densely defined and non-negative definite operator $\mathrm{\hat{H}}:\Phi
\rightarrow \Phi .$

\section{Nonrelativistic many-particle integrable quantum models, their
current algebra representations and the Hamiltonian reconstruction}

\subsection{An integrable many-particle oscillatory quantum model}

As a first application of the local current algebra representation
construction devised above, we consider a simple nonrelativistic
oscillatory quantum model \cite{Aref} of interacting bose-particles in the $%
m $-dimensional Euclidean space $(\mathbb{R}^{m};\langle\cdot |\cdot \rangle),m\in
\mathbb{Z}_{+},$ described by the secondly quantized Hamiltonian operator 
\begin{equation}
\mathrm{H}^{(\omega )}=\frac{1}{2}\int_{\mathbb{R}^{m}}\rd x\langle\nabla \psi
^{+}(x)|\nabla \psi (x)\rangle +\frac{1}{2}\int_{\mathbb{R}^{m}}\rd x\langle\omega x|\omega
x\rangle \psi ^{+}(x)\psi (x),\   \label{G1}
\end{equation}%
acting on the corresponding Fock space $ \Phi , $ and parameterized by
the positive definite frequency matrix $\omega \in \End\mathbb{R}^{m}.$ 
\ It is easy to check that for any $|\Omega _{0})\in \Phi $, the vector
\begin{equation}
|\varphi _{\omega })=\exp \left( -\frac{1}{2}\int_{\mathbb{R}^{m}}\rd x\rho
(x)\langle x|\omega x\rangle \right) |\varphi _{0})  \label{G2}
\end{equation}%
satisfies the identity%
\begin{equation}
\mathrm{K}(x)|\varphi _{\omega })=\exp \left( -\frac{1}{2}\int_{\mathbb{R}%
^{m}}\rd x\rho (x)\langle x|\omega x\rangle \right) \left[ \mathrm{K}(x)+\omega x\rho (x)%
\right] |\varphi _{0}),  \label{G3}
\end{equation}%
from which one easily ensues
\begin{align}
&\left( \mathrm{H}^{(\omega )}-\frac{1}{2}\tr\omega\,\textrm{N}%
\right) |\varphi _{0})
=\frac{1}{2}\int_{\mathbb{R}^{m}}\rd x\langle [\mathrm{K}^{+}(x)+\omega x\rho (x)]|\rho
(x)^{-1}[\mathrm{K}(x)+\omega x\rho (x)]\rangle |\varphi _{0}) \nonumber\\
&=\exp \left( \frac{1}{2}\int_{\mathbb{R}^{m}}\rd x\rho (x)\langle x|\omega x\rangle \right)
\mathrm{H}^{(0)}\exp \left( -\frac{1}{2}\int_{\mathbb{R}^{m}}\rd x\rho
(x)\langle x|\omega x\rangle \right) |\varphi _{0}),  
\label{G4}
\end{align}%
meaning that operators $( \mathrm{H}^{(\omega )}-\frac{1}{2}\tr %
\omega\, \textrm{N}) $ and $\mathrm{H}^{(0)}$ are equivalent in
the Fock space $\Phi _\text{F}.$ Thus, we can formulate the following proposition.

\begin{proposition}
The quantum oscillatory Hamiltonian operator (\ref{G1}) permits on the
suitable Fock space $\Phi _\textup{F}$ the factorized representation
\begin{equation}
\mathrm{\hat{H}}^{(\omega )}=\frac{1}{2}\int_{\mathbb{R}^{m}}\rd x\langle[\mathrm{K}%
^{+}(x)+\omega x\rho (x)]|\rho (x)^{-1}[\mathrm{K}(x)+\omega x\rho (x)]\rangle +%
\frac{1}{2}\tr \omega\, \textup{N.}  \label{G5}
\end{equation}%
Its ground state $|\Omega ^{(\omega )})\in \Phi $ satisfies the conditions%
\begin{equation}
\mathrm{\hat{H}}^{(\omega )}|\Omega ^{(\omega )})=\frac{1}{2}\tr %
\omega\, \textup{N}|\Omega ^{(\omega )}), \qquad [\mathrm{K}(x)+\omega
x\rho (x)]|\Omega ^{(\omega )})=0  \label{G6}
\end{equation}%
for all $x\in \mathbb{R}^{m}.$
\end{proposition}

As a consequence from the relationship (\ref{G6}), one can easily obtain
that the ground state of the Hamiltonian (\ref{G1}) equals
\begin{equation}
|\Omega ^{(\omega )})=\exp \left( -\frac{1}{2}\int_{\mathbb{R}^{m}}\rd x\rho
(x)\langle x|\omega x\rangle \right) |\Omega ^{(0)}),\   \label{G7}
\end{equation}%
where, by definition, $\mathrm{K}(x)|\Omega ^{(0)})=0$ for all $x\in \mathbb{%
R}^{m},$ and the related average energy density $\epsilon _{0}=$ $%
\lim_{N\rightarrow \infty }\frac{N}{2N}\tr \omega =\frac{1}{2}\tr\omega .$ 
Moreover, as for any $x,y\in \mathbb{R}^{m}$, there holds the
equality
\begin{equation}
\lbrack \langle\mathrm{D}^{(\omega ),+}(x)|\rho (x)^{-1}\mathrm{D}^{(\omega
)}(x)\rangle,\langle\mathrm{D}^{(\omega ),+}(y)|\rho (y)^{-1}\mathrm{D}^{(\omega
)}(y)\rangle]=0,  \label{G8}
\end{equation}%
where, by definition, the local operator 
\begin{equation}
\mathrm{D}^{(\omega )}(x):=\mathrm{K}(x)+\omega x\rho (x),
\label{G9}
\end{equation}%
the next operators%
\begin{equation}
\mathrm{\hat{H}}^{(\omega ,p)}=\frac{1}{2}\int_{\mathbb{R}^{m}}\rd x\left( \langle
\mathrm{D}^{(\omega ),+}(x)|\rho (x)^{-1}\mathrm{D}^{(\omega )}(x)\rangle\right)
^{p}  \label{G10}
\end{equation}%
on the Fock space $\Phi $ \textit{a priori} commute to each other,
that is
\begin{equation}
\lbrack \mathrm{\hat{H}}^{(\omega ,p)},\mathrm{\hat{H}}^{(\omega ,q)}]=0
\label{G11}
\end{equation}%
for any integers $p,q\in \mathbb{Z}_{+}.$ Thus, we have stated the following
quantum integrability proposition.

\begin{proposition}
The nonrelativistic oscillatory quantum model (\ref{G1}) of interacting
bose-particles in the $m$-dimensional space $\mathbb{R}^{m}$ possesses a
countable hierarchy of the commuting to each other symmetric operators (%
\ref{G10}) on the suitable Fock space $\Phi _\textup{F}$ and represents a quantum
completely integrable model.
\end{proposition}

The related $N$-particle differential expressions for the operators (\ref%
{G10}) can be calculated as their corresponding projections on the $N$%
-particle Fock subspace $\Phi _{N}^{(s)}\subset \Phi _\text{F},$ $N\in \mathbb{Z}%
_{+}.$ In particular, if a vector $\ |\varphi _{N})\in \Phi _{N}\ $is
representable as
\begin{equation}
|\varphi _{N}):=\int_{\mathbb{R}^{mN}}f_{N}(x_{1},x_{2},\ldots ,x_{N})\prod%
\limits_{j=\overline{1,N}}\rd x_{j}\psi ^{+}(x_{j})|0)  \label{G12}
\end{equation}%
for some coefficient function $f_{N}\in L_{2}^{(s)}(\mathbb{R}^{mN};\mathbb{C%
}),$ then, by definition,
\begin{equation}
\mathrm{\hat{H}}^{(\omega ,p)}|\varphi _{N}):=|\varphi _{N}^{(p)}),
\label{G13}
\end{equation}%
where
\begin{equation}
|\varphi _{N}^{(p)})=\int_{\mathbb{R}^{mN}}(H_{N}^{(\omega
,p)}f_{N})(x_{1},x_{2},\ldots ,x_{N})\prod\limits_{j=\overline{1,N}}\rd x_{j}\psi
^{+}(x_{j})|0)  \label{G14}
\end{equation}%
for a given $p\in \mathbb{Z}_{+}$ and arbitrary finite $N\in \mathbb{Z}_{+}.$

The  differential operators $H_{N}^{(\omega ,p)}:
L_{2}^{(s)}(\mathbb{R}^{mN};\mathbb{C})\rightarrow L_{2}^{(s)}(\mathbb{R}%
^{mN};\mathbb{C}),p$ $\in \mathbb{Z}_{+}$ obtained this way, respectively, also commute
to each other, as this follows from (\ref{G11}), giving rise to the
quantum integrability of the $N$-particle oscillatory Hamiltonian model $%
\mathrm{\hat{H}}_{N}^{(\omega )}=\mathrm{\hat{H}}_{N}^{(\omega ,1)}+\frac{1}{%
2}\tr \omega \,N$ for arbitrary finite $N\in \mathbb{Z}_{+}. $

\subsection{A generalized integrable many-particle oscillatory quantum model}

A generalized quantum oscillatory model of bose-particles in $ \mathbb{R}%
^{m}$ is described by the $N$-particle Hamiltonian operator
\begin{equation}
H_{N}:=\frac{1}{2}\sum_{j=\overline{1,N}}\langle \nabla _{x_{j}}|\nabla _{x_{j}}\rangle +%
\frac{1}{2}\sum_{j,k=\overline{1,N}}\langle \omega _{N}(x_{j}-x_{k})|\omega
_{N}(x_{j}-x_{k})\rangle \   \label{K1}
\end{equation}%
on $L_{2}^{(s)}(\mathbb{R}^{mN};\mathbb{C}),$ parameterized by a positive
definite interaction matrix $\omega _{N}\in \End \mathbb{R}^{m},N\in
\mathbb{Z}_{+}.$ In the case when this interaction matrix depends on the
particle number $N\in \mathbb{Z}_{+}\ $as $\omega _{N}=\bar{\omega}\sqrt{N/2%
\text{\ }}$ for some constant positive definite matrix $\bar{\omega}\in \End \mathbb{R}^{m},$ the corresponding to (\ref{K1}) secondly quantized
Hamiltonian operator is representable as
\begin{align}
\mathrm{H}=\frac{1}{2}\int_{\mathbb{R}^{m}}\rd x\langle \nabla \psi ^{+}(x)|\nabla
\psi (x)\rangle 
+\frac{\mathrm{N}}{4}\int_{\mathbb{R}^{m}\times \mathbb{R}^{m}}\rd x\rd y\psi
^{+}(x)\psi ^{+}(y)\psi (y)\psi (x)\langle \bar{\omega}(x-y)|\bar{\omega}(x-y)\rangle ,%
\label{K2}
\end{align}%
acting on a suitably chosen Fock type representation space $\Phi _\text{F}.$

Consider now a quasi-local operator\ mapping $\ \mathrm{D}(x):\Phi
\rightarrow \Phi ^{m},x\in \mathbb{R}^{m},$ equal to
\begin{equation}
\mathrm{D}(x):=K(x)+\int_{\mathbb{R}^{m}}\rd y\langle \bar{\omega}(x-y):\rho (x)\rho
(y):\,,  \label{K3}
\end{equation}%
and construct the next operator expression:
\begin{equation}
\mathrm{\hat{H}}=\frac{1}{2}\int_{\mathbb{R}^{m}}\langle \mathrm{D}^{+}(x)|\rho
(x)^{-1}\mathrm{D}(x)\rangle .  \label{K4}
\end{equation}%
Then, the following proposition holds.

\begin{proposition}
The operator expression (\ref{K4}) is equivalent on the Fock space $\Phi
_\textup{F}$ to the secondly quantized Hamiltonian operator (\ref{K2}):%
\begin{equation}
\mathrm{\hat{H}}=\mathrm{H}-\frac{\tr \bar{\omega}}{2}\mathrm{%
N}(\mathrm{N}-1).  \label{K5}
\end{equation}
\end{proposition}

\begin{proof}
It is sufficient to successfully calculate the operator expression (\ref{K4}):%
\begin{align}
\mathrm{\hat{H}}&=\frac{1}{2}\int_{\mathbb{R}^{m}}\rd x\langle \nabla \psi
^{+}(x)|\nabla \psi (x)\rangle +\frac{1}{2}\int_{\mathbb{R}^{m\times 2}}\rd x\rd y\langle \nabla
\psi ^{+}(x)\psi (x)|\bar{\omega}(x-y)\rangle\rho (y) \nonumber \\
&+\frac{1}{2}\int_{\mathbb{R}^{m\times 2}}\rd x\rd y\rho (y)\langle \bar{\omega}(x-y)|\psi
^{+}(x)\nabla \psi (x)\rangle  
+\frac{1}{2}\int_{\mathbb{R}^{m\times 3} }\rd x\rd y\rd z\rho (x)\rho (y)\rho (z)\langle %
\bar{\omega}(x-y)|\bar{\omega}(x-z)\rangle  \nonumber  \\
&=\frac{1}{2}\int_{\mathbb{R}^{m}}\rd x\langle \nabla \psi ^{+}(x)|\nabla \psi (x)\rangle +%
\frac{1}{2}\int_{\mathbb{R}^{m\times 2}}\rd x\rd y\langle \nabla (\psi ^{+}(x)\psi (x))|%
\bar{\omega}(x-y)\rangle \rho (y)  \nonumber  \\
&+\frac{\tr \bar{\omega}}{2}\int_{\mathbb{R}^{m}}\rd x\rho (x)+\frac{1}{2}%
\int_{\mathbb{R}^{m\times 3} }\rd x\rd y\rd z\rho (x)\rho (y)\rho (z)\langle \bar{\omega}%
(x-y)|\bar{\omega}(x-z)\rangle \nonumber  \\
&=\frac{1}{2}\int_{\mathbb{R}^{m}}\rd x\langle \nabla \psi ^{+}(x)|\nabla \psi (x)\rangle -%
\frac{\tr \bar{\omega}}{2}\int_{\mathbb{R}^{m\times
2}}\rd x\rd y\rho (x)\rho (y) \nonumber  \\
&+\frac{1}{2}\int_{\mathbb{R}^{m\times 3} }\rd x\rd y\rd z\rho (x)\rho (y)\rho (z)\langle %
\bar{\omega}(x-y)|\bar{\omega}(x-z)\rangle 
=\frac{1}{2}\int_{\mathbb{R}^{m}}\rd x\langle \nabla \psi ^{+}(x)|\nabla \psi (x)\rangle -%
\frac{\tr \bar{\omega}}{2}\mathrm{N}(\mathrm{N-1)} \nonumber  \\
&+\frac{1}{2}\int_{\mathbb{R}^{m\times 3} }\rd x\rd y\rd z\rho (x)\rho (y)\rho (z)\langle %
\bar{\omega}(x-y)|\bar{\omega}(x-z)\rangle .%
\label{K6}
\end{align}%
The latter component of the operator expression (\ref{K6}) should be
symmetrized to be representable as the corresponding part of the\ particle
interaction operator $\mathrm{V}(\rho ):\Phi \rightarrow \Phi :$
\begin{align}
\mathrm{V}(\rho )&:=\frac{1}{2}\int_{\mathbb{R}^{m\times 3}\ }\rd x\rd y\rd z\rho
(x)\rho (y)\rho (z)\langle \bar{\omega}(x-y)|\bar{\omega}(x-z)\rangle \nonumber  \\
&=\frac{1}{2\cdot 3!}\int_{\mathbb{R}^{m\times 3}\ }\rd x\rd y\rd z:\rho (x)\rho
(y)\rho (z):  \lbrack \langle \bar{\omega}(x-y)|\bar{\omega}(x-z)\rangle +\langle \bar{\omega}(x-z)|\bar{%
\omega}(x-y)\rangle \nonumber \\
&+\langle \bar{\omega}(y-z)|\bar{\omega}(y-x)\rangle +\langle \bar{\omega}(y-x)|\bar{\omega}(y-z)\rangle 
+\langle \bar{\omega}(z-x)|\bar{\omega}(z-y)\rangle +\langle \bar{\omega}(z-y)|\bar{\omega}%
(z-x)\rangle ] \nonumber  \\
&=\frac{1}{2\cdot 3!}\int_{\mathbb{R}^{m\times 3}\ }\rd x\rd y\rd z:\rho (x)\rho
(y)\rho (z):  \lbrack \langle \bar{\omega}(x-y)|\bar{\omega}(x-z)\rangle +\langle \bar{\omega}(y-z)|\bar{%
\omega}(y-x)\rangle \nonumber  \\
&+\langle \bar{\omega}(x-z)|\bar{\omega}(x-y)\rangle +\langle \bar{\omega}(z-x)|\bar{\omega}(z-y)\rangle 
+\langle \bar{\omega}(y-x)|\bar{\omega}(y-z)\rangle +\langle \bar{\omega}(z-y)|\bar{\omega}%
(z-x)\rangle ] \nonumber  \\
&=\frac{1}{2\cdot 3!}\int_{\mathbb{R}^{m\times 3}\ }\rd x\rd y\rd z:\rho (x)\rho
(y)\rho (z):[\langle \bar{\omega}(x-y)|\bar{\omega}(x-z)-\bar{\omega}(y-z)\rangle \nonumber  \\
&+\langle \bar{\omega}(x-z)|\bar{\omega}(x-y)-\bar{\omega}(z-y)\rangle +\langle \bar{\omega}(y-x)-%
\bar{\omega}(z-x)|\bar{\omega}(y-z)\rangle ] \nonumber  \\
&=\frac{3}{2\cdot 3!}\left( \int_{\,\mathbb{R}^{m}}\rd z\rho (z)\right) \int_{%
\mathbb{R}^{m\times 3}\ }\rd x\rd y:\rho (x)\rho (y):\langle \bar{\omega}(x-y)|\bar{\omega%
}(x-y)\rangle \nonumber  \\
&=\frac{\mathrm{N}}{4}\int_{\mathbb{R}^{m\times 3}\ }\rd x\rd y:\rho (x)\rho (y):\langle %
\bar{\omega}(x-y)|\bar{\omega}(x-y)\rangle . 
\label{K7}
\end{align}%
Having substituted the interaction potential (\ref{K7}) into the
operator expression (\ref{K6}), we obtain 
\begin{align}
\mathrm{H}&=\frac{1}{2}\int_{\mathbb{R}^{m}}\rd x\langle \nabla \psi ^{+}(x)|\nabla
\psi (x)\rangle -\frac{\tr \bar{\omega}}{2}\mathrm{N}(\mathrm{N-1)}
\nonumber\\
&+\frac{\mathrm{N}}{4}\int_{\mathbb{R}^{m}\times \mathbb{R}^{m}}\rd x\rd y\psi
^{+}(x)\psi ^{+}(y)\psi (y)\psi (x)\langle \bar{\omega}(x-y)|\bar{\omega}(x-y)\rangle ,%
\label{K8}
\end{align}%
that is exactly the operator relationship (\ref{K5}), proving the
proposition.
\end{proof}

\begin{remark}
It is worthy to remark here that owing to its construction, the operator
mappings \linebreak $\langle \mathrm{D}^{+}(x)|\rho (x)^{-1}\mathrm{D}(x)\rangle :\Phi \rightarrow
\Phi ,x\in \mathbb{R}^{m},$  commute to each other, that is
\begin{equation}
\lbrack \langle \mathrm{D}^{+}(x)|\rho (x)^{-1}\mathrm{D}(x)\rangle ,\langle \mathrm{D}%
^{+}(y)|\rho (y)^{-1}\mathrm{D}(y)\rangle ]=0  \label{K9}
\end{equation}%
for any $x,y\in \mathbb{R}^{m}.$ This naturally makes it possible to
construct a countable hierarchy of commuting to each other operators $\ \
\mathrm{\hat{H}}^{(p)}:\Phi \rightarrow \Phi ,p\in \mathbb{Z}_{+},$ where $%
\mathrm{\ }$
\begin{equation}
\mathrm{\hat{H}}^{(p)}:=\int_{\mathbb{R}^{m}}\rd x\left( \langle \mathrm{D}%
^{+}(x)|\rho (x)^{-1}\mathrm{D}(x)\rangle \right) ^{p},  \label{K10}
\end{equation}%
that is
\begin{equation}
\lbrack \mathrm{\hat{H}}^{(p)},\mathrm{\hat{H}}^{(q)}]=0  \label{K11}
\end{equation}%
for all $p,q\in \mathbb{Z}_{+}.$ The latter, in particular, means that our
generalized quantum oscillatory model (\ref{K1}) is also integrable.
\end{remark}

Consider now the ground state vector $|\Omega ^{(\omega )})\in \Phi ,$
satisfying the conditions%
\begin{equation}
\mathrm{\hat{H}}|\Omega ^{(\omega )})=0, \qquad \mathrm{D}(x)|\Omega
^{(\omega )})=0  \label{K12}
\end{equation}%
for any $x\in \mathbb{R}^{m}.$ Then, as easily follows from (\ref{K5}), \
the oscillatory quantum Hamiltonian \ (\ref{K2}) satisfies the constraint
\begin{equation}
(\Omega ^{(\omega )}|\mathrm{H}|\Omega ^{(\omega )})=\frac{\tr \bar{%
\omega}}{2}\left(||\mathrm{N}|\Omega ^{(\omega )})||^{2}-||\mathrm{N}%
^{1/2}|\Omega ^{(\omega )})||^{2}\right)\geqslant 0,  \label{K13}
\end{equation}%
meaning, in part, that the quantum oscillatory Hamiltonian (\ref{K1}) is
bounded from below.

\section{The spatially one-dimensional integrable quantum systems, their
operator symmetries and related quantum current algebra representations}

\subsection{The Calogero-Moser-Sutherland quantum model: the current algebra
representation, the Hamiltonian reconstruction and integrability}

The periodic Calogero-Moser-Sutherland quantum bosonic model on the finite \
interval $[0,l]\simeq $ $\mathbb{R}/[0,l]\mathbb{Z}$ is governed by the $N$%
-particle Hamiltonian%
\begin{equation}
H_{N}:=-\sum_{j=\overline{1,N}}\frac{\partial ^{2}}{\partial x_{j}^{2}}%
+\sum_{j\neq k=\overline{1,N}}\frac{\piup ^{2}\beta (\beta -1)}{l^{2}\sin ^{2}[%
\frac{\piup }{l}(x_{j}-x_{k})]}  \label{C3.1}
\end{equation}%
in the symmetric Hilbert space $L_{2}^{(s)}([0,l]^{N};\mathbb{C}),$ where $%
N\in \mathbb{Z}_{+}$ and $\beta \in \mathbb{R}$ is an interaction parameter.
 As it was stated in  very interesting and highly speculative works \cite%
{LaVi}, there exist linear differential operators
\begin{equation}
\mathcal{D}_{j}:=\frac{\partial }{\partial x_{j}}-\frac{\piup \beta }{l}%
\sum_{k=\overline{1,N},\,k\neq j}\cot\left[\frac{\piup }{l}(x_{j}-x_{k})\right]
\label{C3.2}
\end{equation}%
for $j=\overline{1,N},$ such that the Hamiltonian (\ref{C3.1}) is
factorized as the bounded from below symmetric operator
\begin{equation}
H_{N}=\sum_{j=\overline{1,N}}\mathcal{D}_{j}^{+}\mathcal{D}_{j}+E_{N},
\label{C3.3}
\end{equation}%
where
\begin{equation}
E_{N}=\frac{1}{3}\left( \frac{\piup \beta }{l}\right) ^{2}N(N^{2}-1)
\label{C3.3a}
\end{equation}%
is the groundstate energy of the Hamiltonian operator (\ref{C3.1}), which means that
 there exists such a vector $|\Omega _{N})\in $ $L_{2}^{(s)}([0,l]^{N};%
\mathbb{C}),\ $satisfying for any $ N\in \mathbb{Z}_{+}$ the
eigenfunction condition
\begin{equation}
H_{N}|\Omega _{N})=E_{N}|\Omega _{N})  \label{C3.4}
\end{equation}%
and equals
\begin{equation}
|\Omega _{N})=\prod\limits_{j<k=\overline{1,N}}\left\{ \sin \left[\frac{\piup }{l}%
(x_{j}-x_{k})\right]\right\} ^{\beta }.  \label{C3.5}
\end{equation}%
Being additionally interested in proving the quantum integrability of the
Calogero-Moser-Sutherland model (\ref{C3.1}), we proceed to its
second quantized representation \cite{BeSh,BoBo} and study it by means of
the current algebra representation approach, described above in section \ref%
{Sec_FQ1} and devised before in \cite%
{Gold,Gold-1,GoGrPoSh,GoMeSh-1,GoMeSh-2,Meni,MeSh}.

The secondly quantized form of the Calogero-Moser-Sutherland Hamiltonian
operator (\ref{C3.1}) looks as follows:
\begin{equation}
\mathrm{H}=\int_{0}^{l}\rd x\psi _{x}^{+}(x)\psi _{x}(x)+\left( \frac{\piup }{l}%
\right) ^{2}\beta (\beta -1)\int_{0}^{l}\rd x\int_{0}^{l}\rd y\frac{\psi
^{+}(x)\psi ^{+}(y)\psi (y)\psi (x)}{\sin ^{2}[\frac{\piup }{l}(x-y)]}\,,
\label{C3.6}
\end{equation}%
acting in the corresponding Fock space $\Phi _\text{F}:=\oplus _{n\in \mathbb{Z}%
_{+}}\Phi _{n}^{(s)},\ \Phi _{n}^{(s)}\simeq L_{2}^{(s)}([0,l]^{n};\mathbb{C}%
),n\in \mathbb{Z}_{+}.$ To proceed to the current algebra representation of
the Hamiltonian operator (\ref{C3.6}), it would be useful to recall the
factorized representation (\ref{C3.3}) and preliminarily construct the
following singular Dunkl type \cite{AnAvJe,Dunk,LaVi} symmetrized $N$%
-particle singular differential operator
\begin{align}
D_{N}(x)&:=\sum_{j=\overline{1,N}}\delta (x-x_{j})\frac{\partial }{\partial
x_{j}}  \nonumber \\
&-\frac{\piup \beta }{2l}\sum_{j\neq k=\overline{1,N}}\left\{ \delta
(x-x_{j})\cot\left[\frac{\piup }{l}(x_{j}-x_{k})\right]+\delta (x-x_{k})\cot\left[\frac{\piup }{l}%
(x_{k}-x_{j})\right]\right\}
\label{C3.7}
\end{align}%
in the Hilbert space $L_{2}^{(s)}([0,l]^{N};\mathbb{C}),$ $N\in \mathbb{Z}%
_{+},$ parameterized by a running point $x\in \mathbb{R}/[0,l]\mathbb{Z}.$
The corresponding secondly quantized representation of the operator (\ref%
{C3.7}) looks as follows:
\begin{equation}
\mathrm{D}(x)=\psi ^{+}(x)\psi _{x}(x)-\frac{\piup \beta }{l}\int_{0}^{l}\rd y%
\text{ }\cot\left[\frac{\piup }{l}(x-y)\right]:\psi ^{+}(x)\psi ^{+}(y)\psi (y)\psi (x):
\label{C3.8}
\end{equation}%
for any $x\in \mathbb{R}/[0,l]\mathbb{Z}$, or in the current-density 
operator representation form as
\begin{equation}
\mathrm{D}(x)=K(x)-\frac{\piup \beta }{2l}\int_{0}^{l}\rd y\cot\left[\frac{\piup }{l}(x-y)\right]:\rho
(x)\rho (y):.%
\label{C3.9}
\end{equation}%
Now, based on the operator (\ref{C3.9}), one can formulate the following
proposition, first stated in \cite{Meni}, using a completely different
approach.

\begin{proposition}
The secondly quantized Hamiltonian operator (\ref{C3.6}) in the Fock space
$\Phi _\textup{F}$ is representable in dual to (\ref{C3.3}) factorized form as
\begin{equation}
\mathrm{H}=\int_{0}^{l}\rd x\mathrm{D}^{+}(x)\rho (x)^{-1}\mathrm{D}(x)+\mathrm{%
E}\,,  \label{C3.10}
\end{equation}%
where the ground state energy operator $\mathrm{E}:\Phi \rightarrow \Phi $
equals
\begin{equation}
\mathrm{E}=\frac{1}{3}\left( \frac{\piup \beta }{l}\right) ^{2}:\mathrm{N}%
^{3}:+\left( \frac{\piup \beta }{l}\right) ^{2}:\mathrm{N}^{2}:\,,
\label{C3.11}
\end{equation}%
where
\begin{equation}
\mathrm{N}:=\int_{0}^{l}\rho (x)\rd x  \label{C3.11a}
\end{equation}%
is the particle number operator, and satisfies the determining conditions
\begin{equation}
(\mathrm{H}-\mathrm{E})|\Omega )=0,\qquad \mathrm{D}(x)|\Omega )=0
\label{C3.12}
\end{equation}%
on the suitably renormalized groundstate $|\Omega )\in \Phi $ \ for all $%
x\in \mathbb{R}/[0,l]\mathbb{Z}.$ Moreover, for any integer $N\in \mathbb{Z%
}_{+}$, the corresponding projected vector $|\Omega _{N}):=|\Omega )|_{\Phi
_{N}}$  satisfies the following eigenfunction relationships:%
\begin{align*}
\mathrm{N}|\Omega _{N}) =N|\Omega _{N}), 
\end{align*}
\begin{align}
 \mathrm{E}|\Omega
_{N})&=\left( \frac{1}{3}\left( \frac{\piup \beta }{l}\right) ^{2}:\mathrm{N}%
^{3}:+\left( \frac{\piup \beta }{l}\right) ^{2}:\mathrm{N}^{2}:\right) |\Omega
_{N}) \nonumber \\
&=\left[ \frac{1}{3}\left( \frac{\piup \beta }{l}\right)
^{2}(N^{3}-3N^{2}+2N)+\left( \frac{\piup \beta }{l}\right) ^{2}N(N-1)\right]
|\Omega _{N}) \notag \\
&=\left[ \frac{1}{3}\left( \frac{\piup \beta }{l}\right)
^{2}(N^{3}-3N^{2}+2N+3N^{2}-3N)\right] |\Omega _{N})  \notag \\
&=\left[ \frac{1}{3}\left( \frac{\piup \beta }{l}\right) ^{2}N(N^{2}-1)\right]
|\Omega _{N}):=E_{N}|\Omega _{N}),  
\label{C3.13}
\end{align}%
coinciding exactly with the result (\ref{C3.3a}).
\end{proposition}

\begin{remark}
When deriving the expression (\ref{C3.13}), we have used the identity
\begin{align}
\rho (x)\rho (y)&=\text{ }:\rho (x)\rho (y):+\rho (y)\delta (x-y), \nonumber\\
\rho (x)\rho (y)\rho (z)&=\text{ }:\rho (x)\rho (y)\rho (z):+:\rho (x)\rho
(y):\delta (y-z) \nonumber\\
&+:\rho (y)\rho (z):\delta (z-x)+:\rho (z)\rho (x):\delta (x-y)+:\rho
(x)\delta (y-z)\delta (z-x),%
\label{C3.14}
\end{align}%
which holds \cite{Bere,BoBo,GoGrPoSh,MeSh} for the density operator $\rho
:\Phi \rightarrow \Phi $ at any points $x,y,z\in \mathbb{R}/[0,l]\mathbb{Z}%
. $
\end{remark}

Observe now that the operator (\ref{C3.8}) can be rewritten down as follows:
\begin{equation}
\mathrm{D}(x)=\mathrm{K}(x)-\mathrm{A}(x),  \label{C3.15}
\end{equation}%
where
\begin{align}
\mathrm{K}(x) &:=\nabla _{x}\rho (x)/2+\ri J(x),
  \notag \\
\mathrm{A}(x) &:=\frac{\piup \beta }{l}\int_{0}^{l}\rd y\text{ }\cot\left[\frac{\piup }{l%
}(x-y)\right]:\rho (x)\rho (y): 
\label{C3.16} 
\end{align}%
for all $x\in \mathbb{R}/[0,l]\mathbb{Z}.$ Recalling now the second
condition of (\ref{C3.12}), one can rewrite it equivalently as
\begin{equation}
\mathrm{K}(x)|\Omega )=\mathrm{A}(x)|\Omega )  \label{C3.17}
\end{equation}%
on the renormalized groundstate vector $|\Omega )\in \Phi $ for all $x\in
\mathbb{R}/[0,l]\mathbb{Z}.$ On the other hand, owing to the expression (%
\ref{C3.10}), we obtain the searched  current algebra representation
\begin{equation}
\mathrm{\hat{H}}=\int_{0}^{l}\rd x[\mathrm{K}^{+}(x)-\mathrm{A}(x)]\rho
(x)^{-1}[\mathrm{K}(x)-\mathrm{A}(x)]  \label{C3.18}
\end{equation}%
of the Calogero-Moser-Sutherland Hamiltonian operator (\ref{C3.1}) in a
suitably constructed cyclic Hilbert space $\Phi ,$ as it was already
demonstrated in the work \cite{Meni,MeSh}, using the condition (\ref{C3.17}%
) in the form~(\ref{FQeq2.42}).

Now shortly discuss  the quantum integrability of the
Calogero-Moser-Sutherland model (\ref{C3.1}). Owing to the factorized
representation (\ref{C3.18} one can easily observe that for any integer $%
p\in \mathbb{Z}_{+}$, the operators $\mathrm{D}^{+}(x)\rho (x)^{-1}\mathrm{D%
}(x):\Phi \rightarrow \Phi ,x\in \mathbb{R}/[0,l]\mathbb{Z}$, commute to
each other and with the particle number operator $\mathrm{N}:\Phi
\rightarrow \Phi ,$ that is
\begin{equation}
\lbrack \mathrm{D}^{+}(x)\rho (x)^{-1}\mathrm{D}(x),\mathrm{D}^{+}(y)\rho
(y)^{-1}\mathrm{D}(y)]=0,\qquad [\mathrm{D}^{+}(x)\rho (x)^{-1}\mathrm{D}%
(x),\mathrm{N}]=0  \label{C3.19}
\end{equation}%
for any $x,y\in \mathbb{R}/[0,l]\mathbb{Z}.$ As a result of the commutation
property (\ref{C3.19}), one easily obtains that, for any integer $p\in
\mathbb{Z}_{+}$, the symmetric operators
\begin{equation}
\mathrm{\hat{H}}^{(p)}:=\int_{0}^{l}\rd x\left( \mathrm{D}^{+}(x)\rho (x)^{-1}%
\mathrm{D}(x)\right) ^{p}  \label{C3.20}
\end{equation}%
also commute to each other
\begin{equation}
\lbrack \mathrm{H}^{(p)},\mathrm{\hat{H}}^{(q)}]=0  \label{C3.21}
\end{equation}%
for all integers $p,q\in \mathbb{Z}_{+},$ and in particular, commute to the
Calogero-Moser-Sutherland Hamiltonian operator (\ref{C3.10}):%
\begin{equation}
\lbrack \mathrm{\hat{H}}^{(p)},\mathrm{\hat{H}}]=0.  \label{C3.22}
\end{equation}%
Concerning the related $N$-particle differential expressions for the
operators (\ref{C3.20}), it is sufficient to calculate their projections on
the $N$-particle Fock subspace $\Phi _{N}\subset \Phi _\text{F},$ $N\in \mathbb{Z}%
_{+},$ where one assumes that there exists a suitable embedding $\Phi
\hookrightarrow \Phi _\text{F}.$ Namely, let an arbitrary vector $\ |\varphi
_{N})\in \Phi _{N}$ be representable as
\begin{equation}
|\varphi _{N}):=\int_{[0,l]^{N}}f_{N}(x_{1},x_{2},\ldots ,x_{N})\prod\limits_{j=%
\overline{1,N}}\rd x_{j}\psi ^{+}(x_{j})|0)  \label{C3.23}
\end{equation}%
for some coefficient function $f_{N}\in L_{2}^{(s)}([0,l]^{N};\mathbb{C}).$
Then, by definition,
\begin{equation}
\mathrm{H}^{(p)}|\varphi _{N}):=|\varphi _{N}^{(p)}),  \label{C3.24}
\end{equation}%
where
\begin{equation}
|\varphi
_{N}^{(p)})=\int_{[0,l]^{N}}(H_{N}^{(p)}f_{N})(x_{1},x_{2},\ldots ,x_{N})\prod%
\limits_{j=\overline{1,N}}\rd x_{j}\psi ^{+}(x_{j})|0)  \label{C3.25}
\end{equation}%
for a given $p\in \mathbb{Z}_{+}$ any $N\in \mathbb{Z}_{+}.$ In
particular, for $p=2,$ when $\mathrm{H}^{(2)}+\mathrm{E}=\mathrm{\hat{H}}%
:\Phi \rightarrow \Phi ,$ one easily retrieves the shifted
Calogero-Moser-Sutherland Hamiltonian operator (\ref{C3.1}):%
\begin{equation}
H_{N}^{(2)}=-\sum_{j=\overline{1,N}}\frac{\partial ^{2}}{\partial x_{j}^{2}}%
+\sum_{j\neq k=\overline{1,N}}\frac{\piup ^{2}\beta (\beta -1)}{l^{2}\sin ^{2}[%
\frac{\piup }{l}(x_{j}-x_{k})]}-\left( \frac{\piup \beta }{l}\right) ^{2}\frac{%
N(N^{2}-1)}{3}.  \label{C3.26}
\end{equation}%
Respectively, for higher integers $p>2$, the resulting $N$-particle
differential operator expressions $H_{N}^{(p)}:L_{2}^{(s)}([0,l]^{N};\mathbb{%
C})$ $\rightarrow L_{2}^{(s)}([0,l]^{N};\mathbb{C}),$ $N\in \mathbb{Z}_{+},$
can be obtained the above described way by means of simple yet well
cumbersome calculations, and which will prove to be completely equivalent
to those calculated before in the earlier cited nice work \cite{LaVi}.

\begin{remark}
In the thermodynamical limit, when $\lim_{N\rightarrow \infty ,l\rightarrow
\infty }$ $N/\piup l:=\bar{\rho}>0,$ the structural operator $\mathrm{D}%
(x):\Phi \rightarrow \Phi ,x\in \mathbb{R}/[0,l]\mathbb{Z}$, reduces to
\begin{equation}
\mathrm{\ \bar{D}}(x):=\lim_{N/l\rightarrow \bar{\rho}}\mathrm{D}%
(x)=K(x)-\beta \int_{\mathbb{R}}\rd y\text{ }\frac{:\rho (y)\rho (x):}{x-y}\ ,
\label{C3.27}
\end{equation}%
and respectively, the operator (\ref{C3.1}) reduces to
\begin{equation}
\bar{H}_{N}=-\sum_{j=\overline{1,N}}\frac{\partial ^{2}}{\partial x_{j}^{2}}%
+\beta (\beta -1)\sum_{j\neq k=\overline{1,N}}\frac{\ 1}{(x_{j}-x_{k})^{2}}
\label{C3.28}
\end{equation}%
in the Hilbert space $L_{2}^{(s)}(\mathbb{R}^{N};\mathbb{C})$ for any $N\in
\mathbb{Z}_{+},$ whose secondly quantized operator expression in the
suitable Fock space $\Phi $ equals
\begin{equation}
\mathrm{\bar{H}=}\int_{\mathbb{R}}\rd x\mathrm{\ }\left( \mathrm{\bar{D}}%
^{+}(x)\rho (x)^{-1}\mathrm{\bar{D}}(x)\ +\epsilon _{0}\right) ,
\label{C3.29}
\end{equation}%
where $\epsilon _{0}:=\lim_{N/l\rightarrow \bar{\rho}}\frac{E_{N}}{l}=\bar{%
\rho}^{3}/3$ denotes the average energy density of the reduced
Calogero-Moser-Sutherland Hamiltonian operator (\ref{C3.28}) as $%
N\rightarrow \infty ,$ exactly coinciding with the results obtained earlier 
in \cite{MeSh}.
\end{remark}

\subsection{An integrable many-particle Coulomb type quantum model on axis}

A many particle Coulomb type quantum bose model on the axis is governed by the $%
N $-particle Hamiltonian
\begin{align}
H_{N} &:=-\sum_{j=\overline{1,N}}\frac{\partial ^{2}}{\partial x_{j}^{2}}%
+\sum_{j\neq k=\overline{1,N}}\frac{\alpha }{|x_{j}-x_{k}|}   +\frac{\alpha ^{2}}{3}\sum_{j\neq k\neq s=\overline{1,N}}\left[ \ln
|x_{j}-x_{k}|\frac{(x_{j}-x_{k})(x_{j}-x_{s})}{|x_{j}-x_{k}||x_{j}-x_{s}|}%
\ln |x_{j}-x_{s}|\right. \nonumber \\
&+\ \ln |x_{k}-x_{j}|\frac{(x_{k}-x_{j})(x_{k}-x_{s})}{%
|x_{k}-x_{j}||x_{k}-x_{s}|}\ln |x_{k}-x_{s}|  +\left. \ln |x_{s}-x_{j}|\frac{(x_{s}-x_{j})(x_{s}-x_{k})}{%
|x_{s}-x_{j}||x_{s}-x_{k}|}\ln |x_{s}-x_{k}|\right] 
\label{C1.1}
\end{align}%
acting in the Hilbert space $L_{2}^{(s)}(\mathbb{R}^{N};\mathbb{C}),N\in
\mathbb{Z}_{+},\ $ is parameterized by a real-valued interaction parameter $%
\alpha \in \mathbb{R}$, which modulates both the \textit{%
binary} and \textit{ternary} particle interactions. Its secondly quantized
representation in a suitably chosen Fock space $\Phi _\text{F},$ looks as follows:
\begin{align}
\mathrm{H} &=\int_{\mathbb{R}}\rd x\psi _{x}^{+}(x) |\psi _{x}(x)\rangle+\int_{%
\mathbb{R}^{2}}\rd x\rd y\frac{\alpha }{|x-y|}:\rho (x)\rho (y):   \notag\\
&+\frac{\alpha ^{2}}{3}\int_{\mathbb{R}^{3}}\rd x\rd y\rd z :\rho (x)\rho (y)\rho
(z):\left[ \ln |x-y|\frac{(x-y)(x-z)}{|x-y||x-z|}\ln |x-z|\right.   \notag
\\
&+\left. \ln |y-x|\frac{(y-z)(y-x)}{|y-z||y-x|}\ln |y-x|+\ln |z-x|\frac{%
(z-x)(z-y)}{|z-x||z-y|}\ln |z-y|\right] ,  \label{C1.2}
\end{align}%
modulo the infinite renormalization constant operator, responsible for the
coinciding points $x=y\in \mathbb{R}$ of the Coulomb and logarithmic type
interaction potentials. On the axis $\mathbb{R}$, one can define, for any
points $x\neq y\in $ $\mathbb{R}$ and small positive parameter $%
\varepsilon >0$, the functional expression $s(x-y;\varepsilon ):=\frac{(x-y)}{%
\varepsilon |x-y|^{1-\varepsilon }},$ whose derivative $\partial
s(x-y;\varepsilon )/\partial x=1/|x-y|^{1-\varepsilon }.$ Moreover, if point
$x\rightarrow y\in \mathbb{R},$ one has $\lim_{x\rightarrow
y}s(x-y;\varepsilon )=0$ for any $\varepsilon >0.$

One can observe that for any points $x\neq y\in $ $\mathbb{R}$ the\
Hamiltonian operator (\ref{C1.2}) can be equivalently rewritten as
\begin{align}
\mathrm{H}&=\int_{\mathbb{R}}\rd x\psi _{x}^{+}(x)|\psi _{x}(x)\rangle 
+\frac{\alpha }{2}\lim_{\varepsilon \rightarrow 0}\int_{\mathbb{R}%
^{2}}\rd x\rd y[\partial s(x-y;\varepsilon )/\partial x+\partial s(y-x;\varepsilon
)/\partial y:\rho (x)\rho (y): \nonumber\\
&+\frac{\alpha ^{2}}{3}\lim_{\varepsilon \rightarrow 0} \int_{\mathbb{R}%
^{3}}:\rd x\rd y\rd z\rho (x)\rho (y)\rho (z):[s(x-y;\varepsilon )s(x-z;\varepsilon
) \nonumber\\
&+s(y-z;\varepsilon )s(y-x;\varepsilon )+s(z-x;\varepsilon )s(z-y;\varepsilon
)]\ -\varepsilon ^{-2}\mathrm{N(N}^{2}-1)%
\label{C1.3}
\end{align}%
being a well defined quantum operator in the Fock space $\Phi _\text{F}.$
Introduce now, at any point $x\in $ $\mathbb{R}$, the quasi-local operator
expression
\begin{equation}
\mathrm{D}^{(\varepsilon )}(x):=K(x)-\alpha \int_{\mathbb{R}}\rd y:\rho (x)\rho
(y):s(x-y;\varepsilon ),\   \label{C1.4}
\end{equation}%
acting in the Fock space $\Phi ,$ and construct the following operator:%
\begin{equation}
\mathrm{\hat{H}}^{(\varepsilon )}=\int_{\mathbb{R}}\rd x\langle\mathrm{D}%
^{(\varepsilon ),+}(x)|\rho (x)^{-1}\mathrm{D}^{(\varepsilon )}(x)\rangle.
\label{C1.5}
\end{equation}%
Then, one can formulate the next proposition.

\begin{proposition}
The many-particle Coulomb type Hamiltonian operator (\ref{C1.3}) in a
suitably chosen Fock space $\Phi $ is weakly equivalent, as $\varepsilon
\rightarrow 0,$ to the operator expression (\ref{C1.5}), and satisfies the
following regularized limiting relationship:%
\begin{equation}
\reg\lim_{\varepsilon \rightarrow 0}\mathrm{\hat{H}}^{(\varepsilon
)}:=\lim_{\varepsilon \rightarrow 0}\left[ \mathrm{\tilde{H}}^{(\varepsilon
)}-\frac{\alpha ^{2}}{3\varepsilon ^{2}}\mathrm{N(N}^{2}-1)\right] =\mathrm{%
\hat{H}.}  \label{C1.6}
\end{equation}
\end{proposition}

\begin{proof}
We have from (\ref{C1.5}) that
\begin{align}
\mathrm{\hat{H}}^{(\varepsilon )}&=\int_{\mathbb{R}}\rd x\psi _{x}^{+}(x)\psi
_{x}(x)-\alpha \int_{\mathbb{R}^{2}}\rd x\rd y\rho (y)s(x-y;\varepsilon )\psi
^{+}(x)\psi _{x}(x)  -\alpha \int_{\mathbb{R}^{2}}\rd x\rd y\psi _{x}^{+}(x)\psi (x)\rho
(y)s(x-y;\varepsilon ) \nonumber \\
&+\alpha ^{2}\int_{\mathbb{R}^{3}}\rd x\rd y\rd z\rho (x)\rho (y)\rho
(z)s(x-y;\varepsilon )s(x-z;\varepsilon )=\int_{\mathbb{R}}\rd x\psi
_{x}^{+}(x)\psi (x) \nonumber \\
&+\alpha \int_{\mathbb{R}^{2}}\rd x\rd y:\rho (x)\rho (y):\partial
s(x-y;\varepsilon )/\partial x  +\frac{2\alpha ^{2}}{3!}\int_{\mathbb{R}^{3}}\rd x\rd y\rd z\rho (x)\rho (y)\rho
(z)[s(x-y;\varepsilon )s(x-z;\varepsilon )  \nonumber\\
&+s(y-z;\varepsilon )s(y-x;\varepsilon )+s(z-x;\varepsilon )s(z-y;\varepsilon
)] \nonumber\\
&=\int_{\mathbb{R}}\rd x\psi _{x}^{+}(x)\psi _{x}(x)+\alpha \int_{\mathbb{R}%
^{2}}\rd x\rd y:\rho (x)\rho (y):\frac{1}{|x-y|^{1-\varepsilon }} \nonumber \\
&+\frac{\alpha ^{2}}{3}\int_{\mathbb{R}^{3}}\rd x\rd y\rd z:\rho (x)\rho (y)\rho (z):%
\left[ \ln |x-y|\frac{(x-y)(x-z)}{|x-y||x-y|}\ln |x-z|\right.  \nonumber \\
&+\left. \ln |y-x|\frac{(y-z)(y-x)}{|y-z||y-x|}\ln |y-x|+\ln |z-x|\frac{%
(z-x)(z-y)}{|z-x||z-y|}\ln |z-y|\right]  \nonumber \\
&+\varepsilon ^{-2}\mathrm{N(N}^{2}-1)+O(\varepsilon ).%
\label{C1.7}
\end{align}%
\end{proof}
From the relationship (\ref{C1.7}), one follows that
\begin{equation}
\mathrm{H}^{(\varepsilon )}=\mathrm{\hat{H}}+\frac{\alpha }{3\varepsilon ^{2}%
}\mathrm{N(N}^{2}-1)+O(\varepsilon ),  \label{C1.8}
\end{equation}%
as $\varepsilon \rightarrow 0.$ The latter, evidently, means that the weak
limit (\ref{C1.6}) holds, proving the proposition.

Taking into account that the quasi-operators $\mathrm{D}^{(\varepsilon
),+}(x)\rho (x)^{-1}\mathrm{D}^{(\varepsilon )}(x):\Phi \rightarrow \Phi
,x\in \mathbb{R},$  commute to each other, that is
\begin{equation}
\lbrack \mathrm{D}^{(\varepsilon ),+}(x)\rho (x)^{-1}\mathrm{D}%
^{(\varepsilon )}(x),\mathrm{D}^{(\varepsilon ),+}(y)\rho (y)^{-1}\mathrm{D}%
^{(\varepsilon )}(y)]=0  \label{C1.9}
\end{equation}%
for any points $x,y\in \mathbb{R},$ one can also construct a countable
hierarchy of quantum operators
\begin{equation}
\mathrm{\hat{H}}^{(\varepsilon ,p)}:=\int_{\mathbb{R}}\rd x\left( \mathrm{D}%
^{(\varepsilon ),+}(x)\rho (x)^{-1}\mathrm{D}^{(\varepsilon )}(x)\right)
^{p}\   \label{C1.10}
\end{equation}%
for all $p\in \mathbb{Z}_{+},$ also commuting to each other, that is
\begin{equation}
\lbrack \mathrm{\hat{H}}^{(\varepsilon ,p)},\mathrm{\hat{H}}^{(\varepsilon
,q)}]=0  \label{C1.11}
\end{equation}%
for all $p,q\in \mathbb{Z}_{+}.$ Applying then to the operators (\ref%
{C1.10}) the standard weak regularization scheme, similar to that above, one
can obtain, respectively, a countable hierarchy of quantum operators
\begin{equation}
\reg\lim_{\varepsilon \rightarrow 0}\mathrm{\hat{H}}^{(\varepsilon ,p)}:=%
\mathrm{\hat{H}}^{(p)},  \label{C1.12}
\end{equation}%
commuting to each other, that is
\begin{equation}
\lbrack \mathrm{\hat{H}}^{(p)},\mathrm{\hat{H}}^{(q)}]=0  \label{C1.13}
\end{equation}%
for all $p,q\in \mathbb{Z}_{+}.$ The latter makes it possible to claim that
the many-particle Coulomb type quantum bose model (\ref{C1.1}) on the axis
is completely integrable.

\subsection{Quantum many-particle Hamiltonian dynamical system on the axis with $%
\protect\delta $-inte\-raction, its quantum symmetries and integrability}

In this section, we consider a quantum non-relativistic many-particle
bose-system on the axis $\mathbb{R},$ governed by the Hamiltonian operator:

\begin{equation}
H_{N}:=-\sum_{j=\overline{1,N}}\frac{\partial ^{2}}{\partial x_{j}^{2}}%
+\beta \sum_{j\neq k=\overline{1,N}}\delta (x_{j}-x_{k}),  \label{C2.1}
\end{equation}%
where $\alpha ,\beta \in \mathbb{R}$ are interaction constants, and acting
in the symmetric Hilbert space $L_{2}^{(s)}(\mathbb{R}^{N};\mathbb{C}),$ $%
N\in \mathbb{Z}_{+}.$ The corresponding secondly quantized expression \cite%
{Bere,BlPrSa,BoBo,FaYa,PrMy,Takh} for the Hamiltonian operator (\ref{C2.1}%
) in the related Fock space $\Phi _\text{F}\simeq \sum_{n\in \mathbb{Z}%
_{+}}^{\oplus }L_{2}^{(s)}(\mathbb{R}^{n};\mathbb{C})$ equals
\begin{equation}
\mathrm{H}=\int_{\mathbb{R}}\rd x(\psi _{x}^{+}\psi _{x}+\beta \psi ^{+}\psi
^{+}\psi \psi ),\   \label{C2.2}
\end{equation}%
where the creation $\psi ^{+}$-  and annihilation $\psi $-operators satisfy
the canonical commutator relationships
\begin{align}
\lbrack \psi (x),\psi ^{+}(y)] &=\delta (x-y),   \notag\\
\lbrack \psi ^{+}(x),\psi ^{+}(y)] &=0=[\psi (x),\psi (y)]  \label{C2.3}
\end{align}%
for any $x,y\in \mathbb{R}.$ The Hamiltonian operator (\ref{C2.2}) via the
Heisenberg recipe \cite{BoBo,MiBoPrSa,Takh} naturally generates on the
creation $\psi ^{+}:$ $\Phi _\text{F}\rightarrow \Phi _\text{F}$ and annihilation $%
\psi :\Phi _\text{F}\rightarrow \Phi _\text{F}$ operators the following quantum Schr\"{o}dinger type evolution flow:
\begin{align}
\rd\psi /\rd t &:=\frac{1}{\ri}[\mathrm{H},\psi ]=-\ri\psi _{xx}+2\ri\beta \psi
^{+}\psi ^{2},  \notag \\
\rd\psi ^{+}/\rd t &:=\frac{1}{\ri}[\mathrm{H},\psi ^{+}]=\ri\psi _{xx}^{+}-2\ri\beta
(\psi ^{+})^{2}\psi   \label{C2.4}
\end{align}%
with respect to the temporal parameter $t\in \mathbb{R}.$ Subject to the
quantum Schr\"{o}dinger type evolution flow~(\ref{C2.4}), the particle
number operator $\textrm{N}=\int_{\mathbb{R}}\rho (x)\rd x$ and the
Hamiltonian operator (\ref{C2.2}) in the Fock space $\Phi _\text{F}$ are its
conservative symmetries, that is
\begin{equation}
\frac{\rd}{\rd t}\mathrm{N}=0,\qquad \frac{\rd}{\rd t}\mathrm{H}=0
\label{C2.5}
\end{equation}%
for any $t\in \mathbb{R}.$ The quantum model (\ref{C2.2}), as is well
known \cite{BlPrSa,Skla,Takh}, presents a completely integrable quantum Schr\"{o}dinger type dynamical system, possessing an infinite hierarchy of
quantum commuting to each other operators in the Fock space $\Phi _\text{F}.$
This result was proved by means of the quantum inverse scattering transform
\cite{Skla}, based on the existence of a special so-called Lax type quantum
operator linearization in the associated operator-valued space $C^{\infty
}(\mathbb{R};\End \Phi _\text{F}^{2}).$ In what follows below we 
prove the quantum integrability of the quantum Schr\"{o}dinger type
evolution flow (\ref{C2.4}), making use of the local quantum current
algebra representation technique, devised in \cite%
{Gold,GoGrPoSh,GoMeSh-1,GoMeSh-2,Meni}, similarly the way this was done in
sections above.

Let us define, in the Fock space $\Phi _\text{F}$, the following structural
operator:
\begin{equation}
\mathrm{D}^{(\varepsilon )}(x):=K(x)-\beta \int_{\mathbb{R}}\rd y\vartheta
_{\varepsilon }(x-y):\rho (x)\rho (y):\,,  \label{C2.6}
\end{equation}%
where, for any $\varepsilon >0$, the expression $\vartheta _{\varepsilon
}(x-y):=\vartheta (x-y-\varepsilon )=\{1,$ if $\ x>y-\varepsilon \}\wedge $\
$\{0,$ if $\ x\leqslant y+\varepsilon \}$ for $x,y$ $\in \mathbb{R}$ denotes $\ $%
the shifted classical Heaviside $\vartheta $-function, and construct the
following quantum operator:%
\begin{equation}
\mathrm{\hat{H}}^{(\varepsilon )}:=\int_{\mathbb{R}}\rd x\mathrm{D}%
^{(\varepsilon ),+}(x)\rho (x)^{-1}\mathrm{D}^{(\varepsilon )}(x).\
\label{C2.7}
\end{equation}%
The next proposition (\ref{C2.7}) states an equivalence of the quantum
Hamiltonian operator (\ref{C2.2}) and the weak operator limit $%
\lim_{\varepsilon \rightarrow 0}\mathrm{\hat{H}}^{(\varepsilon )}.$

\begin{proposition}
The many-particle quantum operator (\ref{C2.2}) in a suitably chosen Fock
space $\Phi _\textup{F}$ is weakly equivalent, as $\varepsilon \rightarrow 0,$ to
the operator expression (\ref{C2.7}), and satisfies the following
regularized limiting relationship:%
\begin{equation}
\reg\lim_{\varepsilon \rightarrow 0}\mathrm{H}^{(\varepsilon
)}:=\lim_{\varepsilon \rightarrow 0}\left( \mathrm{H}^{(\varepsilon )}-\beta
^{2}\mathrm{N}^{3}/3\right) =\mathrm{\hat{H}.}  \label{C2.8}
\end{equation}
\end{proposition}

\begin{proof}
Having taken into account that $\vartheta _{\varepsilon }(x-y)\delta
(x-y)=0=\vartheta _{\varepsilon }(x-y)\delta ^{\prime }(x-y)$ and $\vartheta
_{\varepsilon }^{\prime }(x-y)=\delta (x-y-\varepsilon )$ for any $x,y\in
\mathbb{R},$ one can calculate the operator expression (\ref{C2.7}) and
obtain:%
\begin{align}
\mathrm{H}^{(\varepsilon )}&=\int_{\mathbb{R}}\rd x\psi _{x}^{+}(x)\psi
_{x}(x)-\beta \int_{\mathbb{R}^{2}}\rd x\rd y\vartheta _{\varepsilon }(x-y)[\psi
_{x}^{+}(x)\psi (x)\rho (y)+\rho (y)\psi ^{+}(x)\psi _{x}(x)] \nonumber \\
&+\beta ^{2}\int_{\mathbb{R}^{3}}\rd x\rd y\rd z\rho
(x)\rho (y)\rho (z)\vartheta _{\varepsilon }(x-y)\vartheta _{\varepsilon
}(x-z) \nonumber\\
&=\int_{\mathbb{R}}\rd x\psi _{x}^{+}(x)\psi_{x} (x)-\beta \int_{\mathbb{R}%
^{2}}\rd x\rd y\vartheta _{\varepsilon }(x-y)[\psi _{x}^{+}(x)\psi (x)+\psi
^{+}(x)\psi _{x}(x)]\rho (y)\nonumber \\
&-\beta \int_{\mathbb{R}^{2}}\rd x\rd y\vartheta _{\varepsilon }(x-y)[\rho (y),\psi
^{+}(x)\psi _{x}(x)]+\beta ^{2}\int_{\mathbb{R}^{3}}\rd x\rd y\rd z\rho (x)\rho (y)\rho (z)\vartheta
_{\varepsilon }(x-y)\vartheta _{\varepsilon }(x-z)%
\nonumber\\
&=\int_{\mathbb{R}}\rd x\psi _{x}^{+}(x)\psi _{x}(x)+\beta \int_{\mathbb{R}%
^{2}}\rd x\rd y\vartheta _{\varepsilon }^{\prime }(x-y)\rho (x)\rho (y) \nonumber\\
&-\beta \int_{\mathbb{R}^{2}}\rd x\rd y[\vartheta _{\varepsilon }(x-y)\delta
(x-y)\psi ^{+}(x)\psi _{x}(x)-\rho (x)\delta ^{\prime }(x-y)] \nonumber\\
&+\beta ^{2}\int_{\mathbb{R}^{3}}\rd x\rd y\rd z\rho (x)\rho (y)\rho (z)\vartheta
_{\varepsilon }(x-y)\vartheta _{\varepsilon }(x-z) 
=\int_{\mathbb{R}}\rd x\psi _{x}^{+}(x)\psi (x)  \nonumber\\
&+\beta \int_{\mathbb{R}%
^{2}}\rd x\rd y\delta (x-y-\varepsilon )\rho (x)\rho (y) 
+\frac{2\beta ^{2}}{3!}\int_{\mathbb{R}^{3}}\rd x\rd y\rd z\rho (x)\rho (y)\rho
(z)[\vartheta _{\varepsilon }(x-y)\vartheta _{\varepsilon }(x-z)  \nonumber\\
&+\vartheta _{\varepsilon }(y-z)\vartheta _{\varepsilon }(y-x)+\vartheta
_{\varepsilon }(z-x)\vartheta _{\varepsilon }(z-y)] \nonumber\\
&=\int_{\mathbb{R}}\rd x\psi _{x}^{+}(x)\psi_{x} (x)+\beta \int_{\mathbb{R}%
^{2}}\rd x\rho (x)\rho (x-\varepsilon )+\beta ^{2}\mathrm{N}^{3}/3.%
\label{C2.9}
\end{align}
\end{proof}

Insomuch as from the latter expression (\ref{C2.9}) one easily ensues that
\begin{equation}
\lim_{\varepsilon \rightarrow 0}\left( \mathrm{H}^{(\varepsilon )}-\beta ^{2}%
\mathrm{N}^{3}/3\right) =\mathrm{H,}  \label{C2.10}
\end{equation}%
the weak operator relationship (\ref{C2.8}) in the Fock space $\Phi $ is
proved.

It is important to mention now the quasi-local operators $\mathrm{D}%
^{(\varepsilon ),+}(x)\rho (x)^{-1}\mathrm{D}^{(\varepsilon )}(x)$ $:\Phi
\rightarrow \Phi ,$ $x\in \mathbb{R},$ owing to their construction, 
commute to each other, that is
\begin{equation}
\lbrack \mathrm{D}^{(\varepsilon ),+}(x)\rho (x)^{-1}\mathrm{D}%
^{(\varepsilon )}(x),\mathrm{D}^{(\varepsilon ),+}(y)\rho (y)^{-1}\mathrm{D}%
^{(\varepsilon )}(y)]=0  \label{C2.11}
\end{equation}%
for any $x,y\in \mathbb{R}.$ The latter makes it possible to construct a
countable hierarchy of operators
\begin{equation}
\mathrm{\hat{H}}^{(\varepsilon ,p)}:=\int_{\mathbb{R}}\rd x\left( \mathrm{D}%
^{(\varepsilon ),+}(x)\rho (x)^{-1}\mathrm{D}^{(\varepsilon )}(x)\right) ^{p}
\label{C2.12}
\end{equation}%
for $p\in \mathbb{Z}_{+},$ commuting to each other, that is
\begin{equation}
\lbrack \mathrm{\hat{H}}^{(\varepsilon ,p)},\mathrm{\hat{H}}^{(\varepsilon
,q)}]=0  \label{C2.13}
\end{equation}%
for any $p,q\in \mathbb{Z}_{+}.$ Applying to the hierarchy of operators (%
\ref{C2.12}) the standard weak regularization procedure as $\varepsilon
\rightarrow 0,$ one can construct, respectively, a countable hierarchy of
quantum operators
\begin{equation}
\mathrm{\hat{H}}^{(p)}:=\reg\lim_{\varepsilon \rightarrow 0}\mathrm{\hat{H}}%
^{(\varepsilon ,p)}  \label{C2.14}
\end{equation}%
on the Fock space $\Phi $ for $p\in \mathbb{Z}_{+},$ $\ $also commuting to
each other, that is%
\begin{equation}
\lbrack \mathrm{\hat{H}}^{(p)},\mathrm{\hat{H}}^{(q)}]=0  \label{C2.15}
\end{equation}%
for any $p,q\in \mathbb{Z}_{+},$ as this naturally follows from the
commutator relationship (\ref{C2.13}). The latter then means that the
quantum Schr\"{o}dinger type dynamical system (\ref{C2.4}) is integrable, the
other way confirming the classical result of \cite{SkTaFa}.

\begin{remark}
It is worth to mention that the following generalized quantum many-particle
Hamiltonian bose system
\begin{align}
H_{N} &:=-\sum_{j=\overline{1,N}}\frac{\partial ^{2}}{\partial x_{j}^{2}}%
+\alpha \sum_{j\neq k=\overline{1,N}}\delta (x_{j}-x_{k})  \notag\\
&+\ri\beta \sum_{j\neq k=\overline{1,N}}\left( \frac{\partial }{\partial x_{j}%
}\circ \delta (x_{j}-x_{k})\ +\delta (x_{j}-x_{k})\circ \frac{\partial }{%
\partial x_{k}}\right) ,  \label{C2.16}
\end{align}%
where $\alpha ,\beta \in \mathbb{R}$ are interaction constants, and acting
on the symmetric Hilbert space $L_{2}^{(s)}(\mathbb{R}^{N};\mathbb{C}),$ $%
N\in \mathbb{Z}_{+},\ $ (that is with $(\alpha \delta +\beta \delta ^{\prime
})$-interaction potential) generates the corresponding secondly quantized
quantum Hamiltonian system%
\begin{align}
\rd\psi /\rd t &:=\frac{1}{\ri}[\mathrm{H},\psi ]=-\ri\psi _{xx}+2\alpha \psi
^{+}\psi \psi +2\beta \psi ^{+}\psi \psi _{x}\,,   \notag\\
\rd\psi ^{+}/\rd t &:=\frac{1}{\ri}[\mathrm{H},\psi ^{+}]=\ri\psi _{xx}^{+}-2\ri\alpha
(\psi ^{+})^{2}\psi +2\beta \psi _{x}^{+}\psi \psi   \label{C2.16a}
\end{align}%
with the quantum Hamiltonian operator
\begin{equation}
\mathrm{H}=\int_{\mathbb{R}}\rd x[\psi _{x}^{+}\psi _{x}+\alpha \psi ^{+}\psi
^{+}\psi \psi +\ri\beta (\psi ^{+}\psi ^{+}\psi _{x}\psi -\psi _{x}^{+}\psi
^{+}\psi \psi )]  \label{C2.17}
\end{equation}%
on the Fock space ${\Phi }_\text{F},$ which is completely $\ $integrable, as it
was proved\ before in \cite{BlPrSa,MiBoPrSa,PrMy} by means of the quantum
inverse scattering transform. This fact, eventually, allows us to speculate
that there exists a suitable local current algebra cyclic representation
space $\Phi ,$ allowing one to construct a related structural operator $\mathrm{D%
}(x):{\Phi }\rightarrow {\Phi },x\in \mathbb{R},$ factorizing the quantum
Hamiltonian operator $\mathrm{\hat{H}}=\int_{\mathbb{R}}$ $\mathrm{D}%
^{+}(x)\rho (x)^{-1}\mathrm{D}(x),$ $\ $\ and reducing, up to some
renormalizing constant operator, to (\ref{C2.17}) on the corresponding Fock
space~$\Phi _\text{F}.$
\end{remark}

\section{The density functional representation of the local current algebra
symmetry and the factorized structure of quantum integrable many-particle
Hamiltonian systems}

\subsection{The density functional representation of the local current
algebra: the canonical representation reduction}

We are now interested in constructing the density functional representation
of the current algebra~(\ref{FQeq2.37b}) in the Hilbert space $\ \Phi
\simeq \Phi _{\rho }$\ with the cyclic vector $|\Omega )=1\in \Phi .$\ To do
this, let us consider first the ``\textit{creation}'' $a^{+}(x)$ and
``annihilation'' operators $a(x),x\in \mathbb{R}^{m},\ $defined via (\ref%
{FQeq2.37}) in the canonical Fock space $\Phi$, which can be formally
represented as%
\begin{equation}
a^{+}(x)=\sqrt{\rho (x)}\exp [-\ri\vartheta (x)], \qquad a(x)=\exp [\ri\vartheta (x)]%
\sqrt{\rho (x)},  \label{FQeq3.1a}
\end{equation}%
where $\rho (x): \Phi \rightarrow \Phi $ is our density operator and
$\vartheta (x): \Phi \rightarrow \Phi , x\in \mathbb{R}^{m}$, is
some self-adjoint operator. It is important that the operators $\rho (x)$ and $%
\vartheta (x): \Phi \rightarrow \Phi $ realize the canonical
commutation relationships
\begin{align}
\lbrack \rho (x),\rho (y)] &=0=[\vartheta (x),\vartheta (y)],
 \notag\\
\lbrack \rho (y),\vartheta (x)] &=\ri\delta (x-y)  \label{FQeq3.2}
\end{align}%
for any $x,y\in \mathbb{R}^{m}.$ Concerning the current operator $J(x): %
\Phi \rightarrow \Phi , x\in \mathbb{R}^{m},$ one can easily obtain its
equivalent expression
\begin{equation}
J(x)=\rho (x)\nabla \vartheta (x).  \label{FQeq3.3}
\end{equation}%
Based on the canonical relationships (\ref{FQeq3.2}), one can easily
obtain, following \cite{Aref}, that
\begin{equation}
\vartheta (x)=\frac{1}{\ri}\frac{\delta }{\delta \rho (x)}+\ri\sigma \lbrack
\rho (x)],  \label{FQeq3.4}
\end{equation}%
where $\sigma \lbrack \rho (x)]:\Phi \rightarrow \Phi $ is some
function of the density operator $\rho (x):\Phi \rightarrow \Phi ,%
x\in \mathbb{R}^{m}.$ Then, respectively, the current operator (\ref%
{FQeq3.3}) is representable in ${\Phi }$ as follows:
\begin{equation}
J(x)=-\ri\rho (x)\nabla \frac{\delta }{\delta \rho (x)}+\rho (x)\nabla \sigma
\lbrack \rho (x)].  \label{FQeq3.5}
\end{equation}%
The functional-operator expression (\ref{FQeq3.5}) proves to make sense
\cite{Arak,Aref,GoGrPoSh} as operators in the Hilbert space $\Phi _{\rho
}\simeq \Phi $ of $\ $functional valued complex-functions on the manifold $%
\mathcal{M},$ coordinated by the density parameter $\rho :\Phi _{\rho
}\rightarrow \Phi _{\rho }$ and endowed with the scalar product $%
(a|b)_{_{\Phi _{\rho }}}:=\int_{\mathcal{M}}\overline{a(\rho )}b(\rho )\rd\mu
(\rho )\ $subject to some measure $\mu $ on $\mathcal{M}.$ To calculate this
measure $\mu $ on $\mathcal{M},$ we present an explicit isomorphism
between this Hilbert space $\Phi _{\rho }$ and the corresponding Fock space {%
$\Phi _\text{F}$} of spinless bosonic particles in $\mathbb{R}^{m}.$ First, we
determine the support $\supp\mu \subset \mathcal{M}$ of the measure $\mu ,$
having assumed that the manifold
\begin{equation}
\mathcal{M=\cup }_{n\in \mathbb{Z}_{+}}\mathcal{M}_{n}\,,  \label{FQeq3.5a}
\end{equation}%
where $\mathcal{M}_{n}:=\{a(\rho ):\rho (x):=\sum_{j=1}^{n}\delta
(x-c_{j}):a\in C^{\infty }(F^{\prime };\End \Phi _{\rho })\},$ where $%
c_{j}\in \mathbb{R}^{m},j=\overline{1,n},$ $n\in \mathbb{N},$ are arbitrary
vector parameters. The restriction $\rd\mu _{n}$ of the measure $\mu $ on the
submanifold $\mathcal{M}_{n}$ can be presented \cite%
{AlDaKoLy,BeSh,Gold,GoGrPoSh,BoPr} as
\begin{equation}
\rd\mu _{n}=\gamma _{n}(c_{1},c_{2},\ldots ,c_{n})\prod\limits_{j=\overline{1,n}%
}\rd c_{j}\,,  \label{FQeq3.6}
\end{equation}%
where functions $\gamma _{n}:\ \mathbb{R}^{m\times n}\rightarrow \mathbb{R}%
_{+},n\in \mathbb{N},$ should be determined from the condition (\ref%
{FQeq3.5}). In accordance with the manifold structure (\ref{FQeq3.5a}), we
can decompose the Hilbert space $\Phi _{\rho }$ as follows:
\begin{equation}
\Phi _{\rho }=\oplus _{n\in \mathbb{N}}\Phi _{n}\,,  \label{FQeq3.7}
\end{equation}%
where the space $\Phi _{n}$ depends on the mapping $\sigma :\mathcal{%
M\rightarrow }\End\mathcal{(}\Phi _{\rho })$ and consists of functionals that
are bounded on $\mathcal{M}_{n},$ in particular, for any $a(\rho )\in
\mathcal{M}$, the restrictions $a(\rho )|_{\Phi _{n}},n\in \mathbb{N},$
consist of functions of vectors $(c_{1},c_{2},\ldots ,c_{n})\in \mathbb{R}%
^{m\times n},n\in \mathbb{N},$ respectively. The scalar product in $\Phi
_{n},n\in \mathbb{N},$ is suitably defined by means of the expressions (%
\ref{FQeq3.6}). Now, we can construct the isomorphism between the Hilbert
spaces $\Phi _{n},n\in \mathbb{N},$ and the corresponding components $\ {%
\Phi }_{n}^{(s)},n\in \mathbb{N},$ of the related Fock space $ {\Phi }_\text{F}{,%
}$ representing spinless bosonic particles in $\mathbb{R}^{m}.$ In the
Hilbert space $\Phi _{n}:=\Phi _{n}^{(\sigma )},n\in \mathbb{N},$ one can
easily calculate the eigenfunctions $\varphi
_{p_{1},p_{2},\ldots \,,p_{n}}^{(\sigma )}(\rho )\in \Phi _{n}^{(\sigma )}$ of the
free Hamiltonian
\begin{equation}
\mathrm{H}_{0}^{(\sigma )}:=\frac{1}{2}\int_{\mathbb{R}^{m}}\rd x\langle K^{+}(x)|\rho
^{-1}(x)K(x)\rangle   \label{FQeq3.7a}
\end{equation}%
with structural
\begin{equation}
K(x):=\frac{1}{2}\nabla \rho (x)+\ri J^{(\sigma )}(x),\qquad K^{+}(x):=\frac{%
1}{2}\nabla \rho (x)-\ri J^{(\sigma )}(x)  \label{FQeq3.7b}
\end{equation}%
and the momentum 
\begin{equation}
\mathrm{P}^{(\sigma )}:=\int_{\mathbb{R}^{m}}\rd xJ^{(\sigma )}(x)
\label{FQeq3.7c}
\end{equation}%
operators:%
\begin{align}
\mathrm{H}_{0}^{(\sigma )}\varphi _{p_{1},p_{2},\ldots \,,p_{n}}^{(\sigma )}(\rho
) &=\Bigg(\sum_{j=\overline{1,n}}E_{j}\Bigg)\varphi _{p_{1},p_{2},\ldots \,,p_{n}}^{(\sigma
)}(\rho ),  \notag \\
\mathrm{P}^{(\sigma )}\varphi _{p_{1},p_{2},\ldots \,,p_{n}}^{(\sigma )}(\rho )
&=\Bigg(\sum_{j=\overline{1,n}}p_{j}\Bigg)\varphi _{p_{1},p_{2},\ldots \,,p_{n}}^{(\sigma
)}(\rho ),  \label{FQeq3.8}
\end{align}%
where $p_{j}\in \mathbb{R}^{m},j=\overline{1,n}$, are momentums of
bose-particles in $\mathbb{R}^{m},$ \ the operator $\mathrm{H}%
_{0}^{(\sigma )}:\Phi _{\rho }{\rightarrow }$ $\Phi _{\rho }$ is given by
the expressions (\ref{FQeq3.7a}),\ (\ref{FQeq2.36}) and (\ref{FQeq3.5})
and the operator $\mathrm{P}^{(\sigma )}:\Phi _{\rho }{\rightarrow }$ $\Phi
_{\rho }$ is given by the expressions (\ref{FQeq3.7b}) and (\ref{FQeq3.5}),
respectively, within which the current operator $J^{(\sigma )}(x):\Phi
_{\rho }{\rightarrow }$ $\Phi _{\rho }$ is realized under the condition $%
\nabla \sigma \lbrack \rho (x)]:=$ $\sigma \rho (x)^{-1}\nabla \rho (x)$ as follows:
\begin{equation}
J^{(\sigma )}(x)=-\ri\rho (x)\nabla \frac{\delta }{\delta \rho (x)}+\ri\sigma
\nabla \rho (x),  \label{FQeq3.9}
\end{equation}%
where $\sigma \in \mathbb{R}$ is a fixed real-valued parameter. In this case
the eigenfunctions $\varphi _{p_{1},p_{2},\ldots \,,p_{n}}^{(\sigma )}(\rho )\in
\Phi _{n}^{(\sigma )},n\in \mathbb{N},$ can be expressed \cite{Aref,PaScWr}
as
\begin{equation}
\varphi _{p_{1},p_{2},\ldots \,,p_{n}}^{(\sigma )}(\rho )\ =\frac{1}{n!}\bar{%
\varphi}_{0}^{(\sigma )}(\rho )\left(\, \prod\limits_{j=\overline{1,n}%
}B_{p_{j}}(\rho )\cdot 1\right) ,  \label{FQeq3.10}
\end{equation}%
where
\begin{align}
\bar{\varphi}_{0}^{(\sigma )}(\rho ) &:=\exp \left[ (\sigma -1/2)\int_{%
\mathbb{R}^{m}}\rd x\rho (x)\ln \rho (x)\right] ,  \notag\\
B_{p_{j}}(\rho ) &:=\int_{\mathbb{R}^{m}}\rd x\exp (\ri\langle p|x\rangle)\rho (x)\exp \left[-%
\frac{\delta }{\delta \rho (x)}\right].  \label{FQeq3.11} 
\end{align}%
The corresponding $n$-particle Fock subspaces $\ {\Phi }_{n}^{(\sigma
)},n\in \mathbb{N},$ can be naturally represented by means of the vectors
\begin{equation}
|\varphi _{n}^{(\sigma )}):=\frac{1}{\sqrt{n!}}\int_{\mathbb{R}^{m\times
n}}\prod\limits_{j=\overline{1,n}}\rd p_{j}\text{ }f_{n}^{(\sigma
)}(p_{1},p_{2},\ldots ,p_{n})a^{+}(p_{1})a^{+}(p_{2})\ldots a^{+}(p_{n})|0)
\label{FQeq3.12}
\end{equation}%
with functions $f_{n}^{(\sigma )}\in L_{2}^{(s)}(\mathbb{R}^{m\times n};%
\mathbb{C}),n\in \mathbb{N},$ where
\begin{equation}
a^{+}(p):=\frac{1}{(2\piup )^{m/2}}\int_{\mathbb{R}^{m}}\rd x\exp (\ri\langle x|p\rangle)a^{+}(x)
\label{FGeq3.12a}
\end{equation}%
denotes the momentum creation operator for any $\ p\in \mathbb{R}^{m}. $

Moreover, any functional $\varphi _{n}^{(\sigma )}(\rho )\in \Phi
_{n}^{(\sigma )},n\in \mathbb{N},$ can be uniquely represented as
\begin{equation}
\varphi _{n}^{(\sigma )}(\rho ):=\int_{\mathbb{R}^{m\times
n}}\prod\limits_{j=\overline{1,n}}\rd p_{j}f_{n}^{(\sigma
)}(p_{1},p_{2},\ldots ,p_{n})\varphi _{p_{1},p_{2},\ldots \,,p_{n}}^{(\sigma )}(\rho
)\   \label{FQeq3.13}
\end{equation}%
for $f_{n}^{(\sigma )}\in L_{2}^{(s)}(\mathbb{R}^{m\times n};\mathbb{C}),$
since the following condition
\begin{equation}
\left. \left( B_{p_{n+1}}(\rho )\prod\limits_{j=\overline{1,n}%
}B_{p_{j}}(\rho )\cdot 1\right) \right\vert _{\rho =a^{+}(x)a(x)}|\varphi
_{n}^{(\sigma )})=0\   \label{FQeq3.14}
\end{equation}%
holds identically for all $p_{j}\in \mathbb{R}^{m},j=\overline{1,n+1},$ and
arbitrary state $|\varphi _{n}^{(\sigma )})\in \Phi _\text{F},n\in \mathbb{N}.$

\begin{remark}
The condition (\ref{FQeq3.14}) jointly with the constraint $\int_{\mathbb{R%
}^{m}}\rho (x)\rd x=n$  in $ {\Phi }_{n}^{(\sigma )},n\in \mathbb{N},$
should be, in general, naturally satisfied for any current algebra
representation space ${\Phi } $ if and only if  $\rho (x)=\sum_{j=%
\overline{1,n}}\delta (x-c_{j})\in \mathcal{M}_{n}$ for arbitrary $n\in
\mathbb{N}.$
\end{remark}

As a result of the above construction we can state that the Hilbert spaces $%
\Phi _{n}^{(\sigma )},n\in \mathbb{N},$ and Fock subspaces $\Phi
_{n}^{(\sigma )},n\in \mathbb{N},$ are, respectively, isomorphic. As a
consequence, we derive that the Hilbert space $\Phi _{\rho }$ and the Fock
space $\ {\Phi }_\text{F}$ are isomorphic too.

Consider now, following \cite{Aref,GoGrPoSh}, the action of the current
operator (\ref{FQeq3.9}) on the basic vectors $\varphi _{n}^{(\sigma
)}(\rho )\in \Phi _{n}^{(\sigma )},n\in \mathbb{N}$:%
\begin{equation}
J^{(\sigma )}(x)\varphi _{n}^{(\sigma )}(\rho )=\bar{\varphi}_{0}^{\,(\sigma
)}(\rho )\left[-\ri\rho (x)\nabla \frac{\delta }{\delta \rho (x)}+\ri\sigma \nabla
\rho (x)\right]\varphi _{n}^{(\sigma )}(\rho ),  \label{FQeq3.15}
\end{equation}%
from which one ensues easily at $\sigma =1/2$ its $n$-particle
representation on the functional manifold $\mathcal{M}_{n}$:
\begin{align}
\left. J^{(1/2)}(x)\varphi _{n}^{(1/2)}(\rho )\right\vert _{\rho (y)=\sum\limits_{j=%
\overline{1,n}}\delta (y-c_{j})} 
=\sum_{j=\overline{1,n}}\frac{1}{2}[-\ri\delta (x-c_{j})\nabla
_{c_{j}}+\ri\nabla _{c_{j}}\circ \delta (x-c_{j})]\tilde{f}%
_{n}^{\,(1/2)}(c_{1},c_{2},\ldots ,c_{n}),%
\label{FQeq3.16}
\end{align}%
where we took into account that $\bar{\varphi}_{0}^{\,(1/2)}(\rho )=1$ for all
densities $\rho :\Phi _{\rho }\rightarrow \Phi _{\rho }$ and have put, by
definition, the Fourier transform
\begin{equation}
\tilde{f}_{n}^{(1/2)}(c_{1},c_{2},\ldots ,c_{n}):=\int_{\mathbb{R}^{m\times
n}}\prod\limits_{j=\overline{1,n}%
}\rd p_{j}f_{n}^{(1/2)}(p_{1},p_{2},\ldots ,p_{n})\exp \Bigg(\ri\sum_{j=\overline{1,n}%
}\langle p_{j}|c_{j}\rangle\Bigg)   \label{FQeq3.17}
\end{equation}%
for any fixed particle position vectors $c_{j}\in \mathbb{R}^{n},j=\overline{%
1,n},$ and for arbitrary $n\in \mathbb{N}.$ The expression (\ref{FQeq3.16}%
), in particular, means that the current operator $J^{(1/2)}(x):\Phi _{\rho
}\rightarrow \Phi _{\rho }$ is symmetric with respect to the measure $\rd\mu
_{n}^{(1/2)}:=\beta _{n}\prod\nolimits_{j=\overline{1,n}}\rd c_{j}$ on each
functional submanifold $\mathcal{M}^{n}$ for all $n\in \mathbb{N},$ where
the constants $\beta _{n}\in \mathbb{R}_{+},n\in \mathbb{N},$ can be
determined from the normalization condition $||\varphi _{n}^{(1/2)}(\rho
)||_{\Phi _{n}^{(1/2)}}=(\varphi _{n}^{(1/2)}|\varphi _{n}^{(1/2)})_{\Phi
_{n}^{(1/2)}}^{1/2},$ $n\in \mathbb{N}.$ The latter gives rise \cite%
{AlDaKoLy,BeSh,BoPr,Gold} to the following symbolic measure expression
\begin{equation}
\rd\mu _{n}^{(1/2)}:=\ \prod\limits_{x\in \mathbb{R}^{m}}\delta \left( \rho
(x)-\sum_{j=\overline{1,n}}\delta (x-c_{j})\right) \prod\limits_{j=%
\overline{1,n}}\frac{\rd c_{j}}{(2\piup )^{m}}  \label{FQeq3.18}
\end{equation}%
for all $c_{j}\in \mathbb{R}^{n},j=\overline{1,n},\ $and arbitrary $n\in
\mathbb{N}.$

\begin{remark}
As was aptly observed in \cite{Aref}, the choice $\sigma =1/2$ makes it
possible to realize the current algebra representation in the space $%
\mathcal{M}$ of analytic functions, which will be \textit{a priori} assumed
 further, meaning that the corresponding measure can be symbolically expressed
as follows:
\begin{equation}
\rd\mu _{n}^{\ }:=\ \prod\limits_{x\in \mathbb{R}^{m}}\delta \left( \rho
(x)-\sum_{j=\overline{1,n}}\delta (x-c_{j})\right) \prod\limits_{j=%
\overline{1,n}}\frac{\rd c_{j}}{(2\piup )^{n}}  \label{FQeq3.19}
\end{equation}%
on the subspace $\mathcal{M}_{n\text{ }}$ for any $n\in \mathbb{N}.$
\end{remark}

\subsection{The quantum oscillatory model: the density functional current
algebra representation and the Hamiltonian reconstruction}

As a classical application of the above construction, one can consider a
current algebra representation of the quantum Hamiltonian operator \ \ \
\begin{equation}
\mathrm{H}^{(\omega )}=\frac{1}{2}\int_{\mathbb{R}^{m}}\langle K(x)^{+}|\rho
(x)^{-1}K(x)\rangle\rd x+\frac{1}{2}\int_{\mathbb{R}^{m}}\langle \omega x|\omega x\rangle\rho
(x)\rd x\   \label{FQeq4.5}
\end{equation}%
in the corresponding Fock space $\ \Phi $ of the generalized quantum $N$%
-particle oscillatory Hamiltonian
\begin{equation}
H_{N}^{(\omega )}=\frac{1}{2}\sum_{j=\overline{1,N}}\left( \langle \nabla
_{x_{j}}|\nabla _{x_{j}}\rangle+\langle \omega x_{j}|\omega x_{j}\rangle\right) \
\label{FQeq4.5a}
\end{equation}%
for $N\in \mathbb{Z}_{+}$ bose-particles in the $m$-dimensional space $%
\mathbb{R}^{m}$ under the external oscillatory potential, parameterized by
the positive definite frequency matrix $\omega \in \End \mathbb{R}^{m}.$ 

Having shifted the representation Hilbert space $\Phi _{\rho }$ by the
functional $\bar{\varphi}_{0}^{\,(1/2)}(\rho ):=\exp [-\frac{1}{2}\int_{%
\mathbb{R}^{m}}\left\langle x|\omega x\right\rangle $ \linebreak $ \times\rho (x)\rd x]\in \Phi
_{\rho },$ the corresponding current operator (\ref{FQeq3.9}) becomes
\begin{equation}
J^{(\omega )}(x) = -\ri\rho (x)\nabla \frac{\delta }{\delta \rho (x)}+\frac{\ri%
}{2}\nabla \rho (x)-\ri\omega x\rho (x),  \label{FQeq4.5b}
\end{equation}%
entailing simultaneously the related $K$-operator changing \
\begin{equation}
K(x)=\rho (x)\nabla \frac{\delta }{\delta \rho (x)}\rightarrow K^{(\omega
)}(x)=\rho (x)\nabla \frac{\delta }{\delta \rho (x)}+\omega x\rho (x)
\label{FQeq4.5c}
\end{equation}%
for any $x\in \mathbb{R}^{m}.$ The latter gives rise, respectively, to the
following equivalent current algebra functional representation of the
oscillatory Hamiltonian (\ref{FQeq4.5}):
\begin{equation}
\mathrm{\hat{H}}^{(\omega )}=\frac{1}{2}\int_{\mathbb{R}^{m}}\langle K^{(\omega
)}(x)^{+}|\rho (x)^{-1}K^{(\omega )}(x)\rangle\rd x+\frac{1}{2}\tr \omega
\int_{\mathbb{R}^{m}}\rho (x)\rd x,  \label{FQ4.5d}
\end{equation}%
found before in (\ref{G5}) for every $N\in \mathbb{N}.$ The shifted
current operator (\ref{FQeq4.5b}) makes it possible to construct the
suitably deformed free particle measure
\begin{equation}
\rd\mu _{1}^{(\omega )}(\rho ):=\exp \left( -\ \int_{\mathbb{R}^{m}}\rd x\rho
(x)\left\langle x|\omega x\right\rangle \right) \rd\mu _{1}^{(1/2)}(\rho )
\label{FQeq4.5e}
\end{equation}%
on the one-particle functional manifold $\mathcal{M}_{1},$  for which
the following expression
\begin{equation}
(\Omega |\mathrm{\hat{H}}^{(\omega )}\mathrm{|U}(f)|\Omega )=\int_{\mathcal{M%
}}\exp [\ri\rho (f)]\rd\mu _{1}^{(\omega )}(\rho )\   \label{FQeq4.5f}
\end{equation}%
holds for any test function $f\in F.$ \ The latter, jointly with the
related ground state condition $\ |\Omega )=1\in \Phi _{\rho },$ makes it
possible to easily calculate the scalar product elements
\begin{equation}
(\mathrm{U}(f_{1})\Omega |\mathrm{\hat{H}}^{(\omega )}\mathrm{|U}%
(f_{2})|\Omega )=\int_{\mathbb{R}^{m}}\exp [\ri f_{1}(c)+\ri f_{2}(c)]\exp \left(
-\left\langle c|\omega c\right\rangle \right) \frac{\rd c}{(2\piup )^{m}}
\label{FQeq4.6}
\end{equation}%
for any test functions $f_{1},f_{2}\in F.$ The expression (\ref{FQeq4.6})
makes it possible to successfully calculate the matrix elements $(\rho
(f_{p_{1}})\Omega |\mathrm{\hat{H}}^{(\omega )}\mathrm{|}\rho
(f_{p_{2}})|\Omega )\ $\ of the Hamiltonian $\mathrm{\hat{H}}^{(\omega
)}:\Phi _{\rho }\rightarrow \Phi _{\rho }$ on the corresponding eigenvectors
$\ \rho (f_{p})|\Omega )\in \Phi _{\rho }$ for arbitrary $p=p_{1},p_{2}\in
\mathbb{N}\ $and, therefore, to find its spectrum.

Consider now the operator (\ref{FQeq2.36}) taking into account the
analytical current representation (\ref{FQeq3.15}) at $\sigma =1/2$:%
\begin{align}
K(x)\varphi _{n}^{(1/2)}(\rho )&=\left[\rho (x)\nabla \frac{\delta }{\delta \rho
(x)}-1/2\nabla \rho (x)\right]\varphi _{n}^{(1/2)}(\rho )+1/2\nabla \rho (x)\varphi _{n}^{(1/2)}(\rho ) \nonumber \\
&=\rho (x)\nabla \frac{\delta }{%
\delta \rho (x)}\varphi _{n}^{(1/2)}(\rho )%
\label{FQeq4.9}
\end{align}%
for any $n\in \mathbb{N}.$ Having substituted instead of $\varphi
_{n}^{(1/2)}(\rho )\in \Phi _{\rho }^{(n)},n\in \mathbb{N},$ the ground
state eigenfunction $\ \Omega (\rho )=1\in \Phi _{\rho }$, we can easily
retrieve the earlier derived expression (\ref{FQeq2.42}). Moreover, based
on the representation (\ref{FQeq4.5c}) and the definition (\ref{FQeq2.38}%
), one can calculate that
\begin{align}
K^{(\omega )}(x)\bar{\varphi}_{0}^{\,(1/2)}(\rho ) =\left[ \rho (x)\nabla
\frac{\delta }{\delta \rho (x)}+\omega x\rho (x)\right] \bar{\varphi}%
_{0}^{\,(1/2)}(\rho )=0  
=A^{(\omega )}(x;\rho )\bar{\varphi}^{\,(1/2)}(\rho ),  \label{FQeq4.9a}
\end{align}%
where $\bar{\varphi}_{0}^{\,(1/2)}(\rho )=\exp [-\frac{1}{2}\int_{\mathbb{R}%
^{m}}\left\langle x|\omega x\right\rangle \rho (x)\rd x]\in \Phi _{\rho
}^{(1/2)}\simeq \Phi _{\rho }.\ $\ The latter means, in particular, that the
corresponding multiplication operator $A^{(\omega )}(x;\rho )=0,$ or,
respectively,
\begin{equation}
K(x)\bar{\varphi}_{0}^{\,(1/2)}(\rho ):=A(x;\rho )\bar{\varphi}%
_{0}^{\,(1/2)}(\rho )=-\omega x\rho (x)\bar{\varphi}_{0}^{\,(1/2)}(\rho ),
\label{FQeq4.9b}
\end{equation}%
where $\ \bar{\varphi}_{0}^{\,(1/2)}(\rho ):=|\Omega (\rho ))\in \Phi _{\rho }$
is the corresponding ground state vector in $\Phi _{\rho }\ $for the
oscillatory Hamiltonian operator \ (\ref{FQeq4.5a}). Making use of the
operator (\ref{FQeq4.5e}), based on expression (\ref{FQeq2.41a}), one
can present a special solution to the functional equation (\ref{FQeq2.41})
in the form
\begin{equation}
\mathcal{L}(f)=\exp \left( -\int_{\mathbb{R}^{m}}\rd x\left\langle \omega
x|x\right\rangle \frac{1}{2\ri}\frac{\delta }{\delta f(x)}\right) \exp \left(
\bar{\rho}\int_{\mathbb{R}^{m}}\{\exp [\ri f(x)]-1\}\rd x\right) ,
\label{FQeq4.10}
\end{equation}%
confirming similar statements from \cite{GoGrPoSh,GoMeSh-1,MeSh}.

\section{Conclusion}

In the work we succeeded in developing an effective algebraic scheme of
constructing density operator and density functional representations for
the local quantum current algebra and its application to quantum Hamiltonian
and symmetry operators reconstruction. We analyzed the corresponding
factorization structure for quantum Hamiltonian operators, governing
spatially many- and one-dimensional integrable dynamical systems. The
quantum generalized oscillatory, Calogero-Moser-Sutherland and nonlinear Schr\"{o}dinger models of spinless bose-particles were analyzed in detail. The
central vector of the density operator current algebra representation
proved to be the ground vector state of the corresponding completely
integrable factorized quantum Hamiltonian system in the classical Bethe
anzatz form. The latter makes it possible to classify quantum completely
integrable Hamiltonian systems a priori allowing the factorized form and
whose groundstate is of the Bethe anzatz from. These and related aspects of
the factorized and completely integrable quantum Hamiltonian systems are
planned to be studied in other place.

\section{Acknowledgements}

Authors would like to convey their warm thanks to Prof. Gerald A. Goldin for
many discussions of the work and instrumental help in editing a manuscript
during the XXVIII International Workshop on ``Geometry in Physics'', held on
30.06--07.07.2019 in Bia\l owie\.{z}a, Poland. They also are cordially
appreciated to Profs. Joel Lebowitz, Denis Blackmore and Nikolai N.
Bogolubov (Jr.) for instructive discussions, useful comments and remarks on
the work. A special authors' appreciation belongs to Prof.~Joel Lebowitz
for the invitation to take part in the 121-st Statistical Mechanics
Conference, held on May 12--14, 2019 in the Rutgers University, New
Brunswick, NJ, USA. Personal A.P.'s acknowledgement belongs to the
Department of Physics, Mathematics and Computer Science of the Cracow
University of Technology for a local research grant F-2/370/2018/DS.

\ukrainianpart

\title{Представлення алгебри струмів для квантових багаточастинкових гамільтонових систем типу Шредінгера, їх факторизована структура та інтегровність}
\author{Д. Пророк\refaddr{label1}, А.К. Прикарпатський\refaddr{label2}}
\addresses{
\addr{label1} Факультет фізики та прикладної інформатики, Університет науки та технологій,  Краків, Польща
\addr{label2} Факультет фізики, математики та інформатики, Краківський технологічний університет, \\Краків, Польща}

\makeukrtitle

\begin{abstract}
\tolerance=3000%
Розвинута ефективна  схема дослідження  представлень квантової  алгебри струмів та
алгебраїчної реконструкції  факторизованих  квантових операторів  Гамільтона  та їх симетрій у просторі типу Фока.  Її застосовано  до квантових  операторів Гамільтона та їх симетрій  у випадку квантових  інтегровних  просторово одно- та багатовимірних  динамічних систем. В якості прикладів нами  детально вивчено факторизовану структуру  гамільтонових операторів, що описують такі квантові інтегровні просторово багатовимірні та одновимірні моделі як узагальнені осциляторні,   модель Калоджеро-Мозера-Сазерленда, типу Кулона та  нелінійна динамічна система  Шредінгера для системи безспінових бозе-частинок.%
\keywords простір Фока, представлення алгебри струмів, реконструкція гамільтоніана, породжуючий функціонал Боголюбова, модель Калоджеро-Мозера-Сазерленда, квантова інтегровність, квантові симетрії

\end{abstract}

\end{document}